\documentclass[a4paper, 11pt]{article}
\RequirePackage{fullpage}

\usepackage{amsthm,amsmath,amssymb}
\usepackage{subcaption}
\usepackage[usenames,dvipsnames]{xcolor}
\usepackage[pdftex,breaklinks,colorlinks,
    citecolor={BlueViolet}, linkcolor={Blue},urlcolor=Maroon]{hyperref}
\usepackage{tikz}
\usetikzlibrary{arrows.meta}
\usetikzlibrary{quotes}
\usetikzlibrary{calc}
\usepackage{charter,eulervm}%
\usepackage{enumerate}
\usepackage[final,expansion=alltext,protrusion=true]{microtype}

\usepackage{ctable}
\newcommand{\probdef}[4]{
  {\parskip=0pt\par\nopagebreak}
    \begin{center}
    \centering
    \begin{tikzpicture}%
      \node[draw=black!20, rounded corners, inner sep=2ex,text width=#1] {
        \begin{minipage}{\textwidth}%
          \begin{tabular*}{\textwidth}{@{\hspace{.1em}} >{\itshape} p{1.2cm} p{0.85\textwidth} @{}}%
            Input: & #3
            \\
            Output: & #4
          \end{tabular*}%
        \end{minipage}
      };
      \coordinate(x) at (current bounding box.north west);
      \node [draw=white,rectangle,inner sep=3pt,anchor=north west,fill=white] 
      at ($(x)+(6pt,.75em)$) {#2};
    \end{tikzpicture}
    {\par\noindent\ignorespacesafterend}
  \end{center}
}

\theoremstyle{plain} 
\newtheorem{theorem}{Theorem}[section]
\newtheorem{lemma}[theorem]{Lemma}
\newtheorem{corollary}[theorem]{Corollary}
\newtheorem{proposition}[theorem]{Proposition}

\newtheorem{reduction}{Rule}[section]

\tikzstyle{filled vertex}  = [{circle,draw=blue,fill=black!50,inner sep=1pt}]  
\tikzstyle{uvertex} = [{violet, draw, fill=violet!50,inner sep=2pt}]  

\title{Improved Kernels for Edge Modification Problems}
\author{Yixin Cao\thanks{Department of Computing, Hong Kong Polytechnic University, Hong Kong, China.  {\tt yixin.cao@polyu.edu.hk, yuping.ke@connect.polyu.hk}.}
    \and
    Yuping Ke\footnotemark[1] 
}

\begin{document}
\maketitle

\begin{abstract}
  In an edge modification problem, we are asked to modify at most $k$ edges to a given graph to make the graph satisfy a certain property.
  Depending on the operations allowed, we have the completion problems and the edge deletion problems.
  A great amount of efforts have been devoted to understanding the kernelization complexity of these problems.
  We revisit several well-studied edge modification problems, and develop improved kernels for them:
  \begin{itemize}
  \item a $2 k$-vertex kernel for the cluster edge deletion problem, 
  \item a $3 k^2$-vertex kernel for the trivially perfect completion problem,
  \item a $5 k^{1.5}$-vertex kernel for the split completion problem and the split edge deletion problem, and
  \item a $5 k^{1.5}$-vertex kernel for the pseudo-split completion problem and the pseudo-split edge deletion problem.
  \end{itemize}
  Moreover, our kernels for split completion and pseudo-split completion have only $O(k^{2.5})$ edges.
  Our results also include a $2 k$-vertex kernel for the strong triadic closure problem, which is related to cluster edge deletion.
\end{abstract}

\section{Introduction}\label{sec:intro}

In an edge modification problem, we are asked to modify at most $k$ edges to a given graph $G$ to make the graph satisfy a certain property.  In particular, we have  edge deletion problems and completion problems when the allowed operations are edge additions and, respectively, edge deletions.  There is also a more general version that allows both operations.  The present paper will be focused on a single type of modifications.
For most graph properties, these edge modification problems are known to be NP-complete~\cite{yannakakis-81-edge-deletion,sharan-02-thesis,mancini-08-thesis}.
A graph $G$ having a certain property is equivalent to that $G$ belongs to some specific graph class.  Cai~\cite{cai-96-hereditary-graph-modification} observed that if the desired graph class can be characterized by a finite number of forbidden induced subgraphs, then these problems are fixed-parameter tractable. 

One is then naturally interested in the kernelization complexity of edge modification problems toward these \emph{easy} graph classes.  Given an instance $(G, k)$, a {\em kernelization algorithm} produces in polynomial time an equivalent instance $(G', k')$---$(G, k)$ is a yes-instance if and only if $(G', k')$ is a yes-instance---such that $k' \leq k$.  The output instance $(G', k')$ is a \textit{polynomial kernel} if the size of $G'$ is bounded from above by a polynomial function of $k'$.  Although progress has been made in this regard, we get stuck for several important graph classes.
We have evidence that some of them do not have polynomial kernels, under certain complexity assumptions, and it is believed that those that do have are exceptions \cite{marx-20-incompressibility}.  This makes a sharp contrast with the vertex deletion problems (deleting vertices instead of edges), for which a polynomial kernel is guaranteed when the number of forbidden induced subgraphs is finite \cite{flum-grohe-06}.
We refer the reader to the recent survey of Crespelle et al.~\cite{crespelle-20-survey-edge-modification}, particularly its Section 2.1 and Table 1, for the most relevant results.

We revisit several well-studied edge modification problems, and develop improved kernels for them.  Our results are summarized in Table~\ref{table:results}.  All the destination graph classes can be defined by a small number of forbidden induced subgraphs (listed in Figure~\ref{fig:small-graphs}).
It is worth mentioning that the edge deletion problem to a graph class is polynomially equivalent to the completion problem to its complement graph class (consisting of the complements of all graphs in the original graph class).
Moreover, some graph classes, e.g., split graphs ($\{2K_2, C_4, C_5\}$-free), are self-complementary, and thus the edge deletion problem and the completion problem toward such a class are equivalent.

\begin{table}[ht]
  \centering
  \begin{tabular}{l r l }
    \toprule
    problem & previous result & our result
    \\ \midrule
    cluster edge deletion & $4 k$ \cite{gruttemeier-20-strong-triadic-closure} & $2 k$
    \\
    trivially perfect completion & $O(k^7)$ \cite{drange-18-kernel-trivially-perfect} & $3 k^2$
    \\
    split completion (edge deletion) & $O(k^2)$ \cite{ghosh-15-split-completion} & $5 k^{1.5}$
    \\
    pseudo-split completion (edge deletion) & - & $5 k^{1.5}$
    \\ \midrule
    strong triadic closure & $4 k$ \cite{gruttemeier-20-strong-triadic-closure} & $2 k$
    \\
    \bottomrule
  \end{tabular}
  \caption{Main results of this paper, shown as the number of vertices in the kernels.
  }
  \label{table:results}
\end{table}

\begin{figure}[h]
  \centering\small
  \begin{subfigure}[b]{.15\linewidth}
    \centering
    \begin{tikzpicture}[every node/.style={filled vertex},scale=.5]
      \node (a) at (-1,0) {};
      \node (c) at (1,0) {};
      \node (b) at (-1,2) {};
      \draw (b) -- (a) -- (c);
    \end{tikzpicture}
    \caption{$P_3$}
  \end{subfigure}  
  \,
  \begin{subfigure}[b]{.15\linewidth}
    \centering
    \begin{tikzpicture}[every node/.style={filled vertex},scale=.5]
      \node (a) at (-1,0) {};
      \node (c) at (1,0) {};
      \node (b) at (-1,2) {};
      \node (d) at (1,2) {};
      \draw (d) -- (c) (b) -- (a);
    \end{tikzpicture}
    \caption{$2 K_2$}
  \end{subfigure}  
  \,
  \begin{subfigure}[b]{.15\linewidth}
    \centering
    \begin{tikzpicture}[every node/.style={filled vertex},scale=.5]
      \node (a) at (-1,0) {};
      \node (c) at (1,0) {};
      \node (b) at (-1,2) {};
      \node (d) at (1,2) {};
      \draw (b) -- (a) -- (c) -- (d);
    \end{tikzpicture}
    \caption{$P_4$}
  \end{subfigure}  
  \,
  \begin{subfigure}[b]{.15\linewidth}
    \centering
    \begin{tikzpicture}[every node/.style={filled vertex}, scale=.5]
      \node (a) at (-1,0) {};
      \node (c) at (1,0) {};
      \node (b) at (-1,2) {};
      \node (d) at (1,2) {};
      \draw (a) -- (b) -- (d) -- (c) -- (a);
    \end{tikzpicture}
    \caption{$C_4$}
  \end{subfigure}  
  \,
  \begin{subfigure}[b]{.15\linewidth}
    \centering
    \begin{tikzpicture}[every node/.style={filled vertex}, scale=.6]
      \coordinate (v0) at ({90}:1) {};
      \foreach \i in {1,..., 5} {
        \coordinate (v\i) at ({90 - \i * (360 / 5)}:1) {};
        \draw let \n1 = {\i - 1} in (v\n1) -- (v\i);
      }
      \foreach \i in {1,..., 5}
      \node at (v\i) {};
    \end{tikzpicture}
    \caption{$C_5$}
  \end{subfigure}  
  \caption{Forbidden induced graphs.  Note that $2 K_{2}$ and $C_{4}$ are complements of each other, while the complements of $P_{4}$ and $C_{5}$ are themselves.}
  \label{fig:small-graphs}
\end{figure}
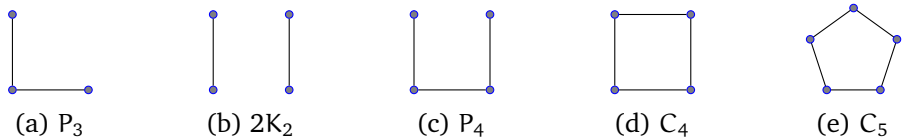

A cluster graph is a disjoint union of cliques.  Since cluster graphs are precisely $P_3$-free graphs, edge modification problems to cluster graphs are the simplest of all nontrivial edge modification problems.  Note that edge modification problems toward $P_2$-free graphs, i.e., edgeless graphs, are trivial.  Also trivial is the cluster completion problem: the minimum solution is to add edges to make every component of the input graph complete.  Both cluster edge editing and cluster edge deletion are NP-complete and have received wide attentions.
After a sequence of results, Cao and Chen~\cite{cao-12-kernel-cluster-editing} devised a $2k$-vertex kernel for the cluster edge editing problem.
Their algorithm actually implies a $2 k$-vertex kernel for the cluster edge deletion problem.  We record this simple result here for future reference.
Less trivially, we show that the same algorithm produces a kernel of the same size for the strong triadic closure problem, which, though originally not posed as an edge modification problem, is closely related to cluster edge deletion \cite{konstantinidis-18-strong-triadic-closure}.
As in \cite{cao-12-kernel-cluster-editing}, both algorithms work for the weighted versions of the problems as well.

The second problem is the trivially perfect completion problem.
Drange and Pilipczuk~\cite{drange-18-kernel-trivially-perfect} presented an $O(k^7)$-vertex kernel for this problem, and they posed as a ``challenging question'' to improve it to $O(k^3)$.  We propose a very simple kernelization algorithm, which has only two simple reduction rules, and the resulting kernel contains at most $2 k^2 + 2 k$ vertices. 
The forbidden induced subgraphs of trivially perfect graphs are $P_4$ and $C_4$.  Note that adding the edge to connect the two ends of a $P_{4}$ merely turns it into a $C_{4}$.  Thus, in each $P_4$ or $C_4$, there are two missing edges such that every solution needs to contain at least one of them.  Note that each vertex of the $P_{4}$ or $C_{4}$ is an end of one of the two missing edges.  Our first rule is the most routine for this kind of problems, namely, adding a missing edge if it is one of the two possible missing edges in $k+1$ or more $P_4$'s and $C_4$'s.  Our second rule removes all vertices that are not contained in any $P_4$ or $C_4$ of $G$.
Now the analysis is similar as Buss and Goldsmith's kernelization algorithm for the vertex cover problem \cite{buss-93-nondeterminism-within-p}.
Since every solution contains at least one of the pair of potential missing edges (for some $P_4$ or $C_4$), and since each potential edge is in at most $k$ pairs, there cannot be more than $k^{2} + k$ potential edges in a yes-instance.  On the other hand, every vertex is in a $P_4$ or $C_4$, hence an end of some potential edge.  We are thus safe to return a trivial no-instance when $|V(G)| > 2 k^{2} + 2 k$.
Toward this result we also obtain some nontrivial observations on minimal solutions of the problem with respect to modules of the input graph.

A graph is a split graph if its vertex set can be partitioned into a clique and an independent set.  Split graphs are $\{2 K_2, C_4, C_5\}$-free graphs.  The split completion problem, which is equivalent to split edge deletion, is NP-complete \cite{natanzon-01-edge-modification}, while somewhat surprisingly, the split edge editing problem can be solved in polynomial time \cite{hammer-81-splittance}.  Guo~\cite{guo-07-kernel-edge-deletion} presented an $O(k^{4})$-vertex kernel for the split completion problem, which was improved to $O(k^{2})$ by Ghosh et al.~\cite{ghosh-15-split-completion}.  For the convenience of presentation, we work on the edge deletion problem.  We consider the partition of the vertex set after applying an optimal solution.   We observe that for most of the vertices we know to which side they have to belong.  It is nevertheless not safe to directly delete these ``decided'' vertices.  We thus work on the annotated version, where we mark certain vertices that have to be in the independent set.
Guo~\cite{guo-07-kernel-edge-deletion} has proved that it is safe to remove a vertex that is not contained in any $2 K_2$, $C_4$, or $C_5$.  We show that a similar rule can be applied to annotated instances, and after its application, there can be at most $O(k^{1.5})$ vertices in a yes-instance.  Finally, a simple step that removes the marks concludes the algorithm.  Our kernel for split completion has only $O(k^{2.5})$ edges.  With minor tweaks, our algorithm produces a kernel of the same size for the pseudo-split ($\{2 K_2, C_4\}$-free graphs) edge deletion problem.  A pseudo-split graph is either a split graph or a split graph plus a $C_{5}$.  The first difficulty toward this adaptation is that it is not always safe to remove vertices not contained in any $2 K_2$ or $C_4$.   We get over this obstacle by observing that we can remove vertices not contained in any $2 K_2$, $C_4$, or $C_5$.  As we recycle the reduction rules for split edge deletion, only the arguments for their safeness need to be slightly revised.

\section{Preliminaries}\label{sec:lbfs}

All graphs discussed in this paper are undirected and simple.  The vertex set and edge set of a graph $G$ are denoted by, respectively, $V(G)$ and $E(G)$.
For a subset $U\subseteq V(G)$, denote by $G[U]$ the subgraph of $G$ induced by $U$, and by $G - U$ the subgraph $G[V(G)\setminus U]$, which is further shortened to $G - v$ when $U = \{v\}$.
The \emph{neighborhood} of a vertex $v$ in $G$, denoted by $N_{G}(v)$, comprises vertices adjacent to $v$, i.e., $N_{G}(v) = \{ u \mid u v\in E(G) \}$, and the \emph{closed neighborhood} of $v$ is $N_{G}[v] = N_{G}(v) \cup \{ v \}$.
The \emph{closed neighborhood} and the \emph{neighborhood} of a set $U\subseteq V(G)$ of vertices are defined as $N_{G}[U] = \bigcup_{v \in U} N_{G}[v]$ and $N_{G}(U) =  N_{G}[U] \setminus U$, respectively.
We may omit the subscript when there is no ambiguity on the graph under discussion.
Two vertices $u$ and $v$ are true twins in $G$ if $N[u] = N[v]$; note that true twins are necessarily adjacent.
A \emph{clique} is a set of pairwise adjacent vertices, and an \emph{independent set} is a set of pairwise nonadjacent vertices.
A graph $G$ is \emph{complete} if $V(G)$ is a clique.
A vertex $v$ is \emph{simplicial} if $N[v]$ is a clique, and a vertex $v$ is \emph{universal} if $N[v] = V(G)$.  An induced path and an induced cycle on $\ell$ vertices are denoted by $P_{\ell}$ and $C_{\ell}$ respectively.

For any two subsets $X, Y\subseteq V(G)$, we use $E(X, Y)$ to denote the set of edges of which one end is in $X$ and the other in $Y$.  Note that we do not require $X$ and $Y$ to be disjoint.  Thus, $E(X, X) = E(G[X])$, i.e., all the edges with both ends in $X$, and $E(X, V(G))$ consists of all the edges with at least one end in $X$.

Let $F$ be a fixed graph.  We say that a graph $G$ is \emph{$F$-free} if $G$ does not contain $F$ as an induced subgraph.  For a set $\mathcal{F}$ of graphs, a graph $G$ is \emph{$\mathcal{F}$-free} if $G$ is $F$-free for every $F\in \mathcal{F}$.  If every $F\in \mathcal{F}$ is minimal, i.e., not containing any $F'\in \mathcal{F}$ as a proper induced subgraph, then the set $\mathcal{F}$ of graphs are the (minimal) \emph{forbidden induced subgraphs} of this class.  See Figure~\ref{fig:small-graphs} for the forbidden induced subgraphs considered in the present paper.  For a set $E'$ of edges disjoint from $E(G)$, we denote by $G + E'$ the graph with vertex set $V(G)$ and edge set $E(G)\cup E'$; for a set $E'\subseteq E(G)$, we denote by $G - E'$ the graph with vertex set $V(G)$ and edge set $E(G)\setminus E'$.
The problems to be studied are formally defined as follows, where $\mathcal{G}$ is a graph class.

\probdef{.82\textwidth}{$\mathcal{G}$ completion}{A graph $G$ and a nonnegative integer $k$.}{Is there a set $E_+$ of at most~$k$ edges such that $G + E_+$  is in $\mathcal{G}$?}

\probdef{.82\textwidth}{$\mathcal{G}$ edge deletion} {A graph $G$ and a nonnegative integer $k$.} {Is there a set $E_-$ of at most~$k$ edges such that $G - E_-$ is in $\mathcal{G}$?}

Since it is always clear from the context what problem we are talking about, when we mention an instance $(G, k)$, we do not always explicitly specify the problem.  We use $\mathrm{opt}(G)$ to denote the size of optimal solutions of $G$ for the optimization version of a certain problem.  Thus, $(G, k)$ is a yes-instance if and only if $\mathrm{opt}(G) \le k$.

For each problem, we apply a sequence of reduction rules.  Each rule transforms an instance $(G, k)$ to a new instance $(G', k')$.  We say that a rule is \emph{safe} if $(G, k)$ is a yes-instance if and only if $(G', k')$ is a yes-instance.
Since all of our reduction rules are very simple and obviously doable in polynomial time, we omit the details of their implementation and the analysis of their running time.

\section{Cluster edge deletion and strong triadic closure}

A graph is a cluster graph if every component of this graph is a complete subgraph.
It is easy to verify that a graph is a cluster graph if and only if it is $P_{3}$-free.
Our first problem is the cluster edge deletion problem.
For a vertex set $U\subseteq V(G)$, we write $d(U) = |E(U, V(G)\setminus U)|$, i.e., the number of edges between $U$ and $V(G)\setminus U$; we write $d(v)$ instead of $d(\{v\})$ for a singleton set.

\begin{reduction}\label{rule:cluster-deletion}
  If there is a simplicial vertex $v$ such that $d(N[v]) \le d(v)$, then remove $N[v]$ and decrease $k$ by $d(N[v])$.
\end{reduction}
\begin{proof}[Safeness of Rule~\ref{rule:cluster-deletion}.]
  We show that $\textrm{opt}(G) = \textrm{opt}(G - N[v]) + d(N[v])$.
  Let $E_{-}$ be an optimal solution to the graph $G$.  We have nothing to show if $N[v]$ makes a separate component of $G - E_{-}$.   In the rest of the proof, $N[v]$ is not a component of $G - E_{-}$.
  Let $G[X]$ denote the component of $G - E_{-}$ that contains $v$.  Since $X$ is a clique and $N[v] \ne X$, we have $X\subset N[v]$.
  In other words, neither $X$ nor $N[v]\setminus X$ is empty.
  Since any induced subgraph of $G - E_{-}$ is a cluster graph, the subset of edges in $E_{-}$ with both ends in $V(G)\setminus N[v]$ is a solution to $G - N[v]$.  Noting that this solution is disjoint from $E(X, V(G)\setminus X)$, we have

  \begin{align}
  \label{eq:lower-bound}
  \textrm{opt}(G) \ge& |E_{-}\cap E(G - N[v])| + d(X) \nonumber 
  \\
  \ge& \textrm{opt}(G - N[v]) + |X| \cdot |N[v]\setminus X| \nonumber
  \\
  \ge& \textrm{opt}(G - N[v]) + |X| + |N[v]\setminus X| - 1
  \\
  =& \textrm{opt}(G - N[v]) + d(v),  \nonumber
\end{align}
where the third inequality holds because both $|X|$ and $|N[v]\setminus X|$ are positive integers.
For any solution $E'_{-}$ of $G - N[v]$, the set $E'_{-}\cup E(N[v], V(G)\setminus N[v])$ is a solution of $G$.  Thus,

\begin{equation}
  \label{eq:upper-bound}
  \textrm{opt}(G) \le
  \textrm{opt}(G - N[v]) + d(N[v]) \le \textrm{opt}(G - N[v]) + d(v).
\end{equation}                     
Therefore, all the inequalities in \eqref{eq:lower-bound} and \eqref{eq:upper-bound} are tight.  In other words, if we remove all the edges between $N[v]$ and $V(G)\setminus N[v]$, and then delete an optimal solution to $G - N[v]$, then we have an optimal solution to the graph $G$.
\end{proof}

A trivial but crucial fact is that a solution $E_-$ has at most $2 |E_-|$ ends.  If a vertex $v$ is not an end of any edge in $E_-$, then $v$ has to be simplicial.

\begin{theorem}\label{thm:cluster}
  There is a $2 k$-vertex kernel for the cluster edge deletion problem.
\end{theorem}
\begin{proof}
  Let $G$ be a graph to which Rule~\ref{rule:cluster-deletion} is not applicable.  We show that if $(G, k)$ is a yes-instance, then $|V(G)| \le 2k$.
  Let $E_{-}$ be an optimal solution to $G$, and let $\{v_1, v_2, \dots, v_r \}$ be the vertices that are not incident to any edge in $E_{-}$; they have to be simplicial.  For $i = 1, \ldots, r$, the set $N[v_{i}]$ forms a component of $G -  E_{-}$.
  Note that for distinct $i, j\in \{1, \ldots, r\}$, the sets $N[v_i]$ and $N[v_j]$ are either the same (when $v_i$ and $v_j$ are true twins) or mutually disjoint: if $N[v_i] \ne N[v_j]$ and there exists $x\in N[v_i] \cap N[v_j]$, then one of $x v_{i}$ and $x v_{j}$ needs to be in $E_{-}$.
  We divide the cost of each edge $u v\in E_{-}$ and assign them to $u$ and $v$ equally.
  For $i = 1, \ldots, r$, the total cost attributed to all the vertices in $N[v_i]$ is $d(N[v_i]) / 2$, because Rule~\ref{rule:cluster-deletion} does not apply to $v_i$.  Each of the vertices not in $\bigcup_{i=1}^{r} N[v_i]$ is an end of at least one edge in $E_{-}$ and therefore bears cost at least $1/2$.  Summing them up, we get a lower bound for the total cost:

  \begin{align*}
    |E_{-}| \ge& {1 \over 2} \sum_{i=1}^{r} d(N[v_i]) + {1 \over 2} |V(G)\setminus \bigcup_{i=1}^{r} N[v_i]|
    \\
    \ge& {1 \over 2} \sum_{i=1}^{r} |N[v_i]| + {1 \over 2} |V(G)\setminus \bigcup_{i=1}^{r} N[v_i]|
    \\
    \ge& {1 \over 2} |V(G)|.
  \end{align*}
  Thus, $|V(G)|/2 \le |E_{-}| \le k$ for a yes-instance, and we can return a trivial no-instance if $|V(G)| > 2 k$.  This concludes the proof.
\end{proof}

Let us mention that the condition of Rule~\ref{rule:cluster-deletion} can be weakened to $d(N[v]) < 2 d(v) - 1$.
We do not prove the stronger statement because it does not improve the analysis of the kernel size, but let us briefly explain why it is true.
The bound $\textrm{opt}(G)\ge \textrm{opt}(G - N[v]) + 2 d(v) - 1$ holds unless $|X| = 1$ or $|N[v] \setminus X| = 1$; see the third inequality of \eqref{eq:lower-bound}.  In the first case, $v$ itself makes a trivial component, and all the vertices in $N(v)$ are in the same component; this can only happen when there exists another vertex $u$ with $N(v) \subseteq N(u)$.  In the second case, a vertex $u\in N(v)$ is incident to all the edges between $N(v)$ and $V(G)\setminus N[v]$.  If $d(N[v]) < 2 d(v) - 1$, then $\textrm{opt}(G)\ge \textrm{opt}(G - N[v]) + 2 d(v) - 1$ holds in both cases.

\begin{figure}[h]
  \centering\small
  \begin{subfigure}[b]{.30\linewidth}
    \centering
    \begin{tikzpicture}[every node/.style={filled vertex}, scale=.75]
      \foreach \i in {1, ..., 4} {
        \node (u\i) at ({90 * \i + 45}:1) {};
        \node (v\i) at ({90 * \i - 135}:2) {};
      }
      \foreach \i in {1, ..., 4} 
      \foreach \j in {1, ..., 4} {
        \ifthenelse{\i=\j}{}{\draw (u\i) -- (v\j)};
        \ifthenelse{\i=\j}{}{\draw (u\i) -- (u\j)};
      }
    \end{tikzpicture}
    \caption{}
  \end{subfigure}  
  \,
  \begin{subfigure}[b]{.30\linewidth}
    \centering
    \begin{tikzpicture}[every node/.style={filled vertex}, scale=.75]
      \foreach \i in {1, ..., 4} {
        \node (u\i) at ({90 * \i + 45}:1) {};
        \node (v\i) at ({90 * \i - 135}:2) {};
      }
      \foreach \j in {1, 2, 4} \draw (v3) -- (u\j);
      \draw (u1) -- (u2) -- (u4) -- (u1);
      \draw (v1) -- (u3);
    \end{tikzpicture}
    \caption{}
  \end{subfigure}  
  \,
  \begin{subfigure}[b]{.30\linewidth}
    \centering
    \begin{tikzpicture}[every node/.style={filled vertex}, scale=.75]
      \foreach \i in {1, ..., 4} {
        \node (u\i) at ({90 * \i + 45}:1) {};
        \node (v\i) at ({90 * \i + 45}:2) {};
        \draw (u\i) -- (v\i);
      }
      \draw (u1) -- (u2) -- (u3) -- (u4) -- (u1);
    \end{tikzpicture}
    \caption{}
  \end{subfigure}  
  \caption{The example given by Konstantinidis et al.~\cite{konstantinidis-18-strong-triadic-closure}: (a) the input graph; (b) a maximum cluster subgraph with seven edges; and (c) a maximum strong triadic closure with eight edges.}
  \label{fig:strong-triadic-closure}
\end{figure}
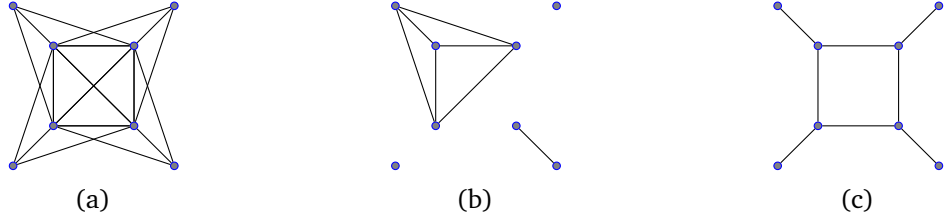

In the original definition, which was motivated by applications in social networks, the \emph{strong triadic closure} problem asks for a partition of the edge set of  the input graph into strong edges and weak ones, such that for every two vertices that are linked to a common neighbor with strong edges are adjacent.  The objective is to maximize the number of strong edges.  For our purpose, it is more convenient to define the problem as follows.

\probdef{.82\textwidth}{Strong triadic closure} {A graph $G$ and a nonnegative integer $k$.} {Is there a set $E_-$ of at most~$k$ edges such that the missing edge of every $P_3$ of $G - E_-$ is in $E(G)$?}

Thus, we call the set of weak edges as the solution to the strong triadic closure problem.
For any set $E_{-}\subseteq E(G)$, if $G - E_{-}$ is a cluster graph, then $E_{-}$ is also a solution to the strong triadic closure problem: setting all edges in $E_{-}$ weak, and all other edges strong is a feasible partition of $E(G)$.
As illustration in Figure~\ref{fig:strong-triadic-closure}, however, a strong triadic closure of a graph can have fewer weak edges than an optimal solution to the cluster edge deletion problem on the same graph.
Surprisingly, Rule~\ref{rule:cluster-deletion} works for the strong triadic closure problem without change.
\begin{lemma}\label{lem:rule:strong-triadic-closure}
  Rule~\ref{rule:cluster-deletion} is safe for the strong triadic closure problem.
\end{lemma}
\begin{proof}
  We show that $\textrm{opt}(G) = \textrm{opt}(G - N[v]) + d(N[v])$.
  Let $E_{-}$ be an optimal solution to the graph $G$.  We have nothing to show if $N[v]$ makes a separate component of $G - E_{-}$.   In the rest of the proof, $N[v]$ is not a component of $G - E_{-}$.
  Let $X$ denote the set of vertices with $N[X] = N[v]$, and $Y\subseteq N[v]$ the ends of these edges in $E(N[v], V(G)\setminus N[v])\setminus E_{-}$ (i.e., edges between $N[v]$ and $V(G)\setminus N[v]$ that are not in $E_{-}$).
  Note that $X\ne \emptyset$ because $v\in X$, and $Y\ne \emptyset$ because $E(N[v], V(G)\setminus N[v])\not\subseteq E_{-}$ (otherwise $N[v]$ is a component of $G - E_{-}$ by the minimality of $E_{-}$).

  By definition, the subset of edges in $E_{-}$ with both ends in $G - N[v]$ is a solution to $G - N[v]$.
  By the selection of $X$ and $Y$, every vertex in $N[v]\setminus (X\cup Y)$ is incident to at least one edge in $E_{-}\cap E(N[v], V(G)\setminus N[v])$.
  For every $x\in X$ and every $y\in Y$, there exists $z\in V(G)\setminus N[v]$ that is adjacent to $y$ but not $x$; hence, $x y z$ is a $P_3$.  As a result, all the edges between $X$ and $Y$ have to be in $E_{-}$.  Thus,

  \begin{align}
    \label{eq:lower-bound-2}
    \textrm{opt}(G) =& |E_{-}\cap E(G - N[v])| + |E_{-}\cap E(N[v], V(G)\setminus N[v])| + |E_{-}\cap E(N[v])|\nonumber
    \\
    \ge& \textrm{opt}(G - N[v]) + |N[v]\setminus (X\cup Y)| + |X|\cdot |Y|\nonumber
    \\
    \ge& \textrm{opt}(G - N[v]) + |N[v]| - |X| - |Y|  + |X| + |Y| - 1
    \\
    \ge& \textrm{opt}(G - N[v]) + |N[v]| - 1 \nonumber 
    \\
    =& \textrm{opt}(G - N[v]) + d(v), \nonumber
  \end{align}
  where $|X|\cdot |Y|\ge |X| + |Y| - 1$ because both $|X|$ and $|Y|$ are positive integers.
  For any solution $E'_{-}$ of $G - N[v]$, the set $E'_{-}\cup E(N[v], V(G)\setminus N[v])$ is a solution of $G$.  Thus,

\begin{equation}
  \label{eq:upper-bound-2}
  \textrm{opt}(G) \le
  \textrm{opt}(G - N[v]) + d(N[v]) \le \textrm{opt}(G - N[v]) + d(v).
\end{equation}                     
Therefore, all the inequalities in \eqref{eq:lower-bound-2} and \eqref{eq:upper-bound-2} are tight.  In other words, if we remove all the edges between $N[v]$ and $V(G)\setminus N[v]$, and then delete an optimal solution to $G - N[v]$, then we have an optimal solution to the graph $G$.
\end{proof}

The proof of the following theorem is a word-by-word copy of that for Theorem~\ref{thm:cluster}, hence omitted.
\begin{theorem}\label{thm:strong-triadic-closure}
  There is a $2 k$-vertex kernel for the strong triadic closure problem.
\end{theorem}

For the strong triadic closure problem, we may alternatively state Rule~\ref{rule:cluster-deletion} as follows.
\begin{reduction}\label{rule:strong-triadic-closure}
  If there is a simplicial vertex $v$ such that $d(N[v]) \le d(v)$, then set all the edges in $G[N[v]]$ strong, set all the edges between $N[v]$ and $V(G)\setminus N[v]$ weak, and delete $N[v]$.
\end{reduction}

We should remark that our kernelization algorithms for the cluster edge deletion problem and the strong triadic closure problem work for the weighted versions as well; see \cite{cao-12-kernel-cluster-editing}.

\section{Trivially perfect completion}

In this section we study the trivially perfect completion problem.  Trivially perfect graphs are $\{P_4, C_4\}$-free graphs.
If there is a pair of adjacent vertices $u, v$ such that neither $N[u]\setminus N[v]$ nor $N[v]\setminus N[u]$ is empty, then they are contained in a $P_{4}$ or $C_{4}$.
Trivially perfect graphs have many nice characterizations.  Here are two of them.
\begin{theorem}[\cite{wolk-62, yan-96-trivially-perfect}]
  \label{thm:trivially-perfect-characterizations}
  The following are equivalent for a graph $H$.
  \begin{enumerate}[i)]
  \item $H$ is a trivially perfect graph.
  \item Every connected induced subgraph of $H$ contains a universal vertex.
  \item For every pair of adjacent vertices $u$ and $v$, one of $N[u]$ and $N[v]$ is a subset of the other.
  \end{enumerate}
\end{theorem}

We say that a trivially perfect graph $\widehat G$ is a \emph{trivially perfect completion} of $G$ if $V(G) = V(\widehat G)$ and $E(G) \subseteq E(\widehat G)$, and it is minimal if there is no other trivially perfect completion $\widehat G'$ of $G$ with $E(G) \subseteq E(\widehat G') \subset E(\widehat G)$.
The following two observations are very simple.
\begin{proposition}\label{lem:universal-vertex}

  Let $H$ be a connected graph, and let $\widehat H$ be a minimal (minimum) trivially perfect completion of $H$.  For any universal vertex $u$ of $\widehat H$, the graph $\widehat H - u$ is a minimal (minimum) trivially perfect completion of $H - u$.
\end{proposition}
\begin{proof}
  Adding $u$ as a universal vertex to any trivially perfect completion of $H - u$, we end with a trivially perfect completion of $H$.
\end{proof}

\begin{lemma}\label{lem:ascendant-tpc}
  Let $\widehat G$ be a minimal trivially perfect completion of a graph $G$, and let $u$ and $v$ be two vertices of $G$.  If $N_{G}[u]\subseteq N_{G}[v]$, then $N_{\widehat G}[u]\subseteq N_{\widehat G}[v]$, 
\end{lemma}
\begin{proof}
  Suppose for contradiction, $N_{\widehat G}[u]\not\subseteq N_{\widehat G}[v]$.
  By Theorem~\ref{thm:trivially-perfect-characterizations}(iii), $N_{\widehat G}[v]\subset N_{\widehat G}[u]$.
Since $N_{G}[u]\subseteq N_{G}[v]\subseteq N_{\widehat G}[v]\subset N_{\widehat G}[u]$, it follows that $x u\not\in E(G)$ for every $x\in N_{\widehat G}[u]\setminus N_{\widehat G}[v]$.  Let $E_{+} = E(\widehat G)\setminus E(G)$.  We consider
    \[
      E_{+}' = E_{+}\setminus \{x u\mid x\in N_{\widehat G}[u]\setminus N_{\widehat G}[v]\} \text{ and } \widehat G' = G + E_{+}'.
  \]
  Then $N_{\widehat G'}[u] = N_{\widehat G'}[v]$.
  Since $\widehat G' - u = \widehat G - u$, it is a trivially perfect graph.  On the other hand, since $u$ and $v$ are true twins of $\widehat G'$, the graph $\widehat G'$ is also a trivially perfect graph.  But since $E(G) \subseteq E(\widehat G') \subset E(\widehat G)$, we have a contradiction to the minimality of $\widehat G$.
\end{proof}

If a vertex $v$ is not contained in any $P_4$ or $C_4$, then for every neighbor $u$ of $v$, one of $N[u]$ and $N[v]$ is a subset of the other.

\begin{lemma}\label{lem:stable-vertex}
  If a vertex $v$ is not contained in any $P_4$ or $C_4$, then $\mathrm{opt}(G - v) = \mathrm{opt}(G)$.
\end{lemma}
\begin{proof}
  It is trivial that $\mathrm{opt}(G - v) \le \mathrm{opt}(G)$.  For the other direction, we show a stronger statement: any minimal solution $E_{+}$ to $G - v$ is also a solution to $G$.  Let $\widehat G = G + E_{+}$.  We verify that $\widehat G$ is a trivially perfect graph by showing that it satisfies Theorem~\ref{thm:trivially-perfect-characterizations}(ii).  If $v$ is universal in $G$, hence also in $\widehat G$, then we are done; otherwise we show that $\widehat G$ is a trivially perfect graph if and only if a proper induced subgraph of $\widehat G$ is.

  If $G$ is not connected, then we can consider the only component that contains $v$.
If $G$ is connected but $G - v$ is not, then $v$ has to be universal in $G$; otherwise, there is a $P_{4}$ containing $v$.
  In the last and the general case, $v$ is not universal in $G$ and $G - v$ is connected.  We argue that at least one vertex in $N_{G}(v) $ is universal in $\widehat G - v$.
  Let $u$ be any universal vertex of $\widehat G - v$.
  We are done if $u\in N_{G}(v)$, and henceforth we assume that  $u\not\in N_{G}(v)$.  Since $v$ is not in any $P_{4}$, the distance between $v$ and $u$ in $G$ is at most two.  Let $u'\in N_G(u)\cap N_G(v)$.  From that $v$ is not in any $P_4$ or $C_4$ it can be inferred that $N_G[u]\subset N_G[u']$ and $N_{G-v}[u]\subseteq N_{G-v}[u']$.  By Lemma~\ref{lem:ascendant-tpc}, $N_{\widehat G - v}[u]\subseteq N_{\widehat G - v}[u']$, and hence $u'$ is also universal in $\widehat G - v$.
  In either case, we have found a vertex $x\in N_{G}(v)$ that is universal in $\widehat G - v$.
  By Proposition~\ref{lem:universal-vertex}, $\widehat G - \{x, v\}$ is a minimal trivially perfect completion of $G - \{x, v\}$.
  Since the graph is finite, the claim follows.
\end{proof}

As a simple result of Lemma~\ref{lem:stable-vertex}, we have the following reduction rule.  In particular, all universal vertices of every component of $G$ can be removed.

\begin{reduction}\label{rule:safe-vertex-tpg}
  If there is a vertex $v$ that is not contained in any $P_4$ or $C_4$, then remove $v$.
\end{reduction}

For each induced $4$-path or $4$-cycle $v_{1}v_{2}v_{3}v_{4}$, we call the missing edges $\{v_{1}, v_{3}\}$ and $\{v_{2}, v_{4}\}$ the \emph{candidate edges} for this path or cycle.  Clearly, any solution of a graph $G$ contains at least one candidate edge of every $P_4$ or $C_4$; note that a $P_4$ has another missing edge, the addition of which merely turns the $P_4$ into a $C_4$.

\begin{reduction}\label{rule:add-edge}
  If $u v$ is a candidate edge of $k+1$ or more $P_{4}$'s and $C_{4}$'s in $G$, then add the edge $u v$ and decrease $k$ by one.
\end{reduction}
\begin{proof}[Safeness of Rule~\ref{rule:add-edge}.]
  Since each $P_4$ or $C_4$ of $G$ has precisely two candidate edges, if a solution $E_{+}$ of $G$ does not contain $u v$, then $E_{+}$ must contain the other candidate edge of each of the $k + 1$ $P_{4}$'s and $C_{4}$'s, hence $|E_{+}| > k$.
\end{proof}

We are thus ready for the main result of this section.

\begin{theorem}\label{thm:trivially-perfect}
  There is a $(2k^2 + 2 k)$-vertex kernel for the trivially perfect completion problem.
\end{theorem}
\begin{proof}
  After applying Rule~\ref{rule:add-edge} and then Rule~\ref{rule:safe-vertex-tpg} exhaustively, we return $(G, k)$ if $|V(G)| \le 2 k^{2} + 2 k$, or a trivial no-instance otherwise.
  We consider all the candidate edges of $G$.  We say that two candidate edges are associated if they belong to the same $P_{4}$ or $C_{4}$; i.e., their ends are disjoint and together induce a $P_{4}$ or $C_{4}$.  Since Rule~\ref{rule:add-edge} is not applicable, each candidate edge is associated with at most $k$ candidate edges.  On the other hand, of any two associated edges, one has to be in any solution of $G$.  Thus, if $(G, k)$ is a yes-instance, there can be at most $k^{2} + k$ candidate edges.  Since Rule~\ref{rule:safe-vertex-tpg} is not applicable, every vertex is in some $P_{4}$ or $C_{4}$, and hence is an end of a candidate edge.  Thus, $|V(G)| \le 2 k^{2} + 2k$ if $(G, k)$ is a yes-instance.
\end{proof}

The analysis of the kernel in Theorem~\ref{thm:trivially-perfect} is essentially the same as Buss and Goldsmith's kernelization algorithm  for the vertex cover problem \cite{buss-93-nondeterminism-within-p}.  In a sense, we are looking for a vertex cover of an auxiliary graph in which each vertex corresponds to a candidate edge of $G$, and two vertices are adjacent if their corresponding edges are associated.
We note that the same approach implies a simple $O(k^{2})$-vertex kernel for the threshold completion problem, matching the result of Drange et al.~\cite{drange-15-threshold-of-intractability}.  The forbidden induced subgraphs of threshold graphs are $2 K_{2}$, $P_{4}$, and $C_{4}$.  The observation on the missing edges of a $P_{4}$ or $C_{4}$ is the same as above, while the four missing edges of a $2 K_{2}$ can be organized as two pairs such that each solution has to contain at least one from each pair.  However, we are not able to employ the $2k$-vertex kernels for vertex cover to directly derive a linear-vertex kernel for either of the two problems.    

Before closing this section, let us mention some observations that might be of independent interest.  The first is a simple corollary of Lemma~\ref{lem:ascendant-tpc}.

\begin{corollary}\label{lem:clique-tpc}
  If two vertices are true twins of a graph $G$, then they remain true twins of any minimal trivially perfect completion of $G$.
\end{corollary}

A set $M$ of vertices is a module if $N(M) = N(v)\setminus M$ for every $v\in M$.  For example, a set of true twins is a module.  Corollary~\ref{lem:clique-tpc} can be generalized to modules.  For the last lemma, we use the fact that trivially perfect graphs are intersection graphs of nested intervals.  (It can also be derived using the characterization by forbidden induced graphs.)  A set of intervals representing an interval graph $G$ is called an \emph{interval representation} for $G$, where the interval for a vertex $v$ is $I(v)$.

\begin{lemma}\label{lem:module}
  A module $M$ of a graph $G$ remains a module in any minimal trivially perfect completion $\widehat G$ of $G$.
\end{lemma}
\begin{proof}
  Let $E_{+} = E(\widehat G)\setminus E(G)$.  The claim follows from Corollary~\ref{lem:clique-tpc} when $\widehat G[M]$ is a clique: $\widehat G$ is also a minimal trivially perfect completion of $G'$, where $G'$ is the graph obtained from $G$ by adding edges to make $M$ a clique.  In the rest $M$ is not a clique of $\widehat G$, hence not a clique of $G$.

Suppose for contradiction that $M$ is not a module of $\widehat G$.
Let $U$ be the set of common neighbors of $M$ in $\widehat G$.  Since $\widehat G[M]$ is not a clique, $\widehat G[U]$ must be a clique.  Moreover, $U\subset N_{\widehat G}(M)$.
We take the leftmost endpoint $\ell$ and the rightmost endpoint $r$ of $\bigcup_{v\in M} I(v)$.  Note that $[\ell, r] \subseteq I(v)$ for every $v\in U$, and
$I(x)\subseteq I(v)$ for every $x\in N_{\widehat G}(M)\setminus U$ and every $v\in U$.
  Let $p = r - \ell$, and we revise the intervals as follows.  We increase each endpoint $\ge r$ by $p + 2$; and for each vertex $x\in N_{\widehat G}(M)\setminus U$, we set $I(x)$ to be $I(x) + p + 1$.  (Informally speaking, we slide intervals for $N_{\widehat G}(M)\setminus U$ to the right so that they are disjoint from those for $M$.)  We consider the graph $G'$ represented by the revised intervals.  It is easy to verify that these interval are still nested, and $E(\widehat G)\setminus E(G')$ is precisely the set of edges between $M$ and $N_{\widehat G}(M)\setminus U$.  Since $M$ is a module of $G$, we have $N_{G}(M)\subseteq U$.  Thus, $E(G)\subseteq E(G')\subset E(\widehat G)$, which contradicts the minimality of $\widehat G$.  This concludes the proof.
\end{proof}

\section{Split edge deletion and split completion}
\label{sec:split}

A graph is a \emph{split graph} if its vertex set can be partitioned into a clique and an independent set.  We use $C\uplus I$, where $C$ being a clique and $I$ an independent set, to denote a \emph{split partition} of a split graph.  Note that a split graph may have more than one split partition; e.g, a complete graph on ${n}$ vertices has $n + 1$ different split partitions.
The forbidden induced subgraphs of split graphs are $2K_2$, $C_4$, and $C_5$.
From both the definition and the forbidden induced subgraphs we can see that the complement of a split graph is also a split graph.  Thus, the split completion problem is polynomially equivalent to the split edge deletion problem.  For the convenience of presentation, we work on the edge deletion problem.

Note that $(G, k)$ is a yes-instance if and only if there exists a partition $C\uplus I$ of $V(G)$ such that $C$ is a clique and $|E(I, I)| \le k$; this is a split partition of $G - E(I, I)$.  We call such a partition a \emph{valid partition} of the instance $(G, k)$.
The problem is thus equivalent to finding a valid partition.
We notice that some vertices can be easily decided to which side of a valid partition they should belong.  For example, unless the instance is trivial, a simplicial vertex always belong to the independent  set in any valid partition.
Even after we know the destinations of these vertices, however, we cannot safely delete them.
This brings us to the \emph{annotated version} of the problem, where we mark certain vertices that can only be put into the independent set in a valid partition.  We use $(G, I_{0}, k)$ to denote such an annotated instance, where $I_{0}$ denotes the set of marked vertices.  The original instance can be viewed as $(G, \emptyset, k)$, and a valid partition of an annotated instance $(G, I_{0}, k)$ needs to satisfy the additional requirement that $I_{0}\subseteq I$.

We can easily retrieve back an unannotated instance from an annotated instance.  It suffices to add a small number of new vertices and make each of them adjacent to all other vertices but $I_{0}$.
\begin{reduction}\label{rule:clean-up}
  Let $(G, I_{0}, k)$ be an annotated instance.
  Add a clique of $\sqrt{2 k} + 1$ new vertices, and make each of them adjacent to all the vertices in $V(G)\setminus I_{0}$.  Return the result as an unannotated instance.
\end{reduction}
\begin{proof}[Safeness of Rule~\ref{rule:clean-up}.]
  Let $K$ denote the clique of new vertices, and let $(G', k)$ be the resulting instance.
  For any valid partition $C\uplus I$ of $(G, I_{0}, k)$, the partition $(C\cup K)\uplus I$ is a valid partition of $(G', k)$ because $C\subseteq V(G)\setminus I \subseteq N(x)$ for every $x\in K$.
  For a valid partition $C\uplus I$ of $(G', k)$, if any vertex in $I_{0}$ is in $C$, then we must have $K\subseteq I$.  Since $K$ is a clique of order $\sqrt{2 k} + 1$, we have $|E(I, I)| > k$, which contradicts the validity of the partition.
\end{proof}

The aforementioned observation on simplicial vertices is formalized by the following rule.
\begin{reduction}\label{rule:simplicial}
  Let $v$ be a simplicial vertex in $V(G)\setminus I_{0}$.  If $|E(G - (N[v]\setminus I_{0}) )| \le k$, then return a trivial yes-instance.  Otherwise, add $v$ to $I_{0}$.
\end{reduction}
\begin{proof}[Safeness of Rule~\ref{rule:simplicial}.]
  In the first case, $(N[v]\setminus I_{0}) \uplus (V(G)\setminus N[v] \cup I_{0})$ is a valid partition.
  Otherwise, we show by contradiction that $v\in I$ in any valid partition $C\uplus I$ of $(G, I_{0}, k)$.
  Since $C$ is a clique, if $v\in C$, then $C\subseteq N[v]\setminus I_{0}$.  Thus, $E(G - (N[v]\setminus I_{0}))\subseteq E(I, I)$, but then $|E(I, I)| > k$, contradicting the validity of the partition.
\end{proof}

We construct a modulator $M$ as follows. We greedily find a maximal packing of vertex-disjoint $2 K_{2}$'s, $C_{4}$'s, and $C_{5}$'s.  Let $M$ be the set of vertices in all subgraphs we found.
We can terminate the algorithm by returning a trivial no-instance if we have found more than $k$ vertex-disjoint forbidden induced subgraphs from $G$.  Henceforth, we may assume that $|M|\leq 5k$, and we fix a split partition $C_{M}\uplus I_{M}$ of $G-M$.  
The following simple observation enables us to  know the destinations of more vertices.

\begin{lemma}\label{lem:partition}
  For any valid partition $C\uplus I$ of $(G, k)$, if one exists, 
  \begin{enumerate}[i)]
  \item $|I_{M}\cap C| \le 1$; and
  \item $|C_{M}\cap I| \le \sqrt{2 k}$.
  \end{enumerate}
\end{lemma}
\begin{proof}
The first assertion follows from that $I_{M}$ is an independent set and $C$ is a clique.  The second assertion holds because 
\[
  k \ge |E(I, I)| \ge |E(C_{M}\cap I, C_{M}\cap I)| = {|C_{M}\cap I|\choose 2}. \qedhere
\]
\end{proof}

We say that a vertex is a \emph{c-vertex}, respectively, an \emph{i-vertex}, if it is in $C$, respectively, in $I$, for any valid partition $C\uplus I$ of $(G, I_{0}, k)$.
Clearly, every vertex that has more than $k+1$ neighbors in $I_{M}$ is a c-vertex, while the following are i-vertices: 
\begin{itemize}
\item every vertex with more than $\sqrt{2k}$ non-neighbors in $C_{M}$; and
\item every vertex nonadjacent to a c-vertex.
\end{itemize}
We can indeed delete all the c-vertices, as long as we keep their non-neighbors marked.  Note that after obtaining the initial split partition $C_{M}\uplus I_{M}$ of $G - M$, we do not need to maintain the invariant that $M$ is a modulator, though we do maintain that $C_{M}$ is a clique and that $I_{M}$ is an independent set throughout.
During our algorithm, we maintain $M$, $C_{M}$, $I_{M}$, and $I_{0}$ as a partition of $V(G)$.
 Therefore, whenever we mark a vertex, we remove it from the set that originally contains it, and move it to $I_{0}$.
\begin{reduction}\label{rule:decided-vertices}
  Let $(G, I_{0}, k)$ be an annotated instance.
  \begin{enumerate}[i)]
  \item Mark every vertex that has more than $\sqrt{2 k}$ non-neighbors in $C_{M}$.
  \item  If a vertex $v$ has more than $k+1$ neighbors in $I_{M}\cup I_{0}$, then mark every vertex in $V(G)\setminus N[v]$ and delete $v$.
  \end{enumerate}
\end{reduction}
\begin{proof}[Safeness of Rule~\ref{rule:decided-vertices}.]
  Let $I_{0}'$ denote the set of marked vertices after the reduction.  It is trivial that if the resulting instance of i) is a yes-instance, then the original is also a yes-instance.  For ii), any valid partition $C'\uplus I'$ of $(G - v, I_{0}', k)$ can be extended to a valid partition $(C'\cup \{v\})\uplus I'$ of $(G, I_{0}, k)$ because $C'\subseteq V(G)\setminus I_{0}'\subseteq N[v]$.
  
  For the other direction, let $C\uplus I$ be any valid partition of $(G, I_{0}, k)$.
  i)  Since $C$ is a clique, $C_{M}\setminus N(v) \subseteq I$ for every $v\in C$.
  By Lemma~\ref{lem:partition}(ii),
  if $|C_{M}\setminus N(v)| > \sqrt{2 k}$ for some vertex $v$, then $v$ has to be in $I$.  Thus, $C\uplus I$ is also a valid partition of the new instance $(G, I_{0}', k)$.
  ii) By Lemma~\ref{lem:partition}(i), $|I_{M}\setminus I|\le 1$.  As $I_{0}\subseteq I$ and $|N(v)\cap (I_{M}\cup I_{0})| > k + 1$, there are at least $k + 1$ edges between $v$ and $I$. 
  Since $|E(I, I)| \le k$, we must have $v\in C$.  Moreover, since $C$ is a clique, $C\subseteq N[v]$, and every vertex nonadjacent to $v$ has to be in $I$.  This justifies the marking of $V(G)\setminus N[v]$.  Clearly, $(C\setminus \{v\})\uplus I$ is a valid partition of $(G - v, I_{0}', k)$.
\end{proof}

The next rule is straightforward: since $I_{0}$ has to be in the independent set, every solution contains all the edges in $E(I_{0}, I_{0})$.

\begin{reduction}\label{rule:i-edges}
  Let $(G, I_{0}, k)$ be an annotated instance.
  Remove all the edges in $E(I_{0}, I_{0})$, and decrease $k$ accordingly.    
\end{reduction}
\begin{proof}[Safeness of Rule~\ref{rule:i-edges}.]
By the definition of the annotated instance, any solution $E_{-}$ of $(G, I_{0}, k)$ contains all the edges in $E(I_{0}, I_{0})$.  Moreover, $E_{-}\setminus E(I_{0}, I_{0})$ is a solution to $G - E(I_{0}, I_{0})$, and its size is at most $k - |E(I_{0}, I_{0})|$.  On the other hand, if $(G - E(I_{0}, I_{0}), k - |E(I_{0}, I_{0})|)$ is a yes-instance, then any solution of this instance, together with $E(I_{0}, I_{0})$, makes a solution of $(G, I_{0}, k)$ of size at most $k$.
\end{proof}

Once there are no edges among vertices in $I_{0}$, we can replace $I_{0}$ with another independent set as long as we keep track of the number of edges between every vertex $v\in V(G)\setminus I_{0}$ and $I_{0}$.  The following rule reduces the cardinality of $I_{0}$.  Note that if Rule~\ref{rule:decided-vertices} is not applicable, then $p \le k$.

\begin{reduction}\label{rule:merge-I0}
  Let $(G, I_{0}, k)$ be an annotated instance where $I_{0}$ is an independent set.
  Introduce $p$ new vertices $v_{1}$, $v_{2}$, $\ldots$, $v_{p}$, where $p = \max_{v\in V(G)}|N(v)\cap I_{0}|$.  For each vertex $x\in N(I_{0})$, make $x$ adjacent to $v_{1}$, $\ldots$, $v_{|N(x)\cap I_{0}|}$.  Remove all vertices in $I_{0}$, and mark the set of new vertices.
\end{reduction}

Instead of proving the safeness of Rule~\ref{rule:merge-I0}, we prove a stronger statement.  
\begin{lemma}\label{lem:i-vertices} 
  Let $(G, I_{0}, k)$ and $(G', I_{0}', k)$ be two annotated instances where $G - I_{0} = G' - I_{0}'$ and both $I_{0}$ and $I_{0}'$ are independent sets.  If $|N_{G}(x)\cap I_{0}| = |N_{G'}(x)\cap I_{0}'|$ for every $x\in V(G)\setminus I_{0}$, then $(G, I_{0}, k)$ is a yes-instance if and only if $(G', I_{0}', k)$ is a yes-instance.
\end{lemma}
\begin{proof}
  We show that $C\uplus I$ is a valid partition of $(G, I_{0}, k)$ if and only if $C\uplus ((I\setminus I_{0})\cup I_{0}')$ is a valid partition of $(G', I_{0}', k)$.
  Note that
  \[
    |E(I\setminus I_{0}, I_{0})| = \sum_{x\in I\setminus I_{0}} |N_{G}(x)\cap I_{0}| =
    \sum_{x\in I\setminus I_{0}'} |N_{G'}(x)\cap I_{0}'|
    = |E(I\setminus I_{0}', I_{0}')|.
  \]
  Since $G - I_{0} = G' - I_{0}'$, and since there is no edge in $G[I_{0}]$ or $G'[I'_{0}]$, the claim follows.
\end{proof}

Let us
recall an important observation of Guo~\cite{guo-07-kernel-edge-deletion}.

\begin{lemma}[\cite{guo-07-kernel-edge-deletion}]
  \label{lem:safe-vertex-split}
  If a vertex $v$ is not contained in any $2K_2$, $C_4$, or $C_5$, then $\mathrm{opt}(G - v) = \mathrm{opt}(G)$.
\end{lemma}

Both Guo~\cite{guo-07-kernel-edge-deletion} and Ghosh et al.~\cite{ghosh-15-split-completion} used a rule derived from this observation to delete vertices, and this is their only rule that removes vertices from the graph.
We may show that the same rule indeed works for our annotated instances, for which however we have to go through the original argument of \cite{guo-07-kernel-edge-deletion}.
We note that if a vertex $v$ in $I_{0}$ is adjacent to two vertices $u$ and $w$ with $u w \not \in E(G)$, then any solution has to contain at least one of edges $u v$ and $v w$ ($u$ and $w$ cannot be both in the clique).  We say that an induced $P_{3}$ is \emph{$I_{0}$-centered} if the degree-two vertex of this $P_{3}$ is from $I_{0}$.  In a sense, $I_{0}$-centered $P_{3}$'s are ``minimal forbidden structures'' for our annotated instances.  Accordingly, a $C_{4}$ or $C_{5}$ involving a vertex from $I_{0}$ is no longer minimal.  In summary, the ``minimal forbidden structures'' are $C_{4}$'s and $C_{5}$'s in $G - I_{0}$, all $2 K_{2}$'s, and $I_{0}$-centered $P_{3}$'s.  Note that a ``minimal forbidden structure'' intersecting $I_{0}$ has to be a $2 K_{2}$ or an $I_{0}$-centered $P_{3}$, and this gives another explanation of the correctness of Lemma~\ref{lem:i-vertices}, which exchanges these two kinds of  ``minimal forbidden structures'' with each  other.
The following rule can be viewed as the annotated version of the rule of Guo~\cite{guo-07-kernel-edge-deletion}, and its safeness can be argued using Lemma~\ref{lem:safe-vertex-split}.
\begin{reduction}\label{rule:alternative}
  Let $(G, I_{0}, k)$ be an annotated instance where $I_{0}$ is an independent set, and let $v$ be a vertex in $C_{M}$.  
  If $v$ is not contained in any $2K_2$ or any $I_{0}$-centered $P_{3}$, and every $C_4$ and $C_5$ that contains $v$ intersects $I_{0}$, then remove $v$ from $G$.
\end{reduction}
\begin{proof}[Safeness of Rule~\ref{rule:alternative}.]
  We show that $(G, I_{0}, k)$ is a yes-instance if and only if $(G - v, I_{0}, k)$ is a yes-instance, by establishing a sequence of equivalent instances.
  For each edge $x y\in E(G)$ with $x\in I_{0}$ and $y\in V(G)\setminus I_{0}$, introduce a new vertex $v_{x y}$ and make it adjacent to $y$.  Remove all vertices in $I_{0}$, and let $I_{0}'$ denote the set of new vertices.  Let $(G', I_{0}', k)$ denote the resulting instance.  The equivalence between $(G', I_{0}', k)$ and $(G, I_{0}, k)$ follows from Lemma~\ref{lem:i-vertices}.  Then let $(G'', k)$ denote the graph obtained by applying Rule~\ref{rule:clean-up} to $(G', I_{0}', k)$, with $K$ being the added clique.

  We argue that $v$ is not contained in any $2K_2$, $C_4$, or $C_5$ of $G''$.
  Suppose for contradiction that there is a set $F\subseteq V(G'')$ that contains $v$ and induces a $2K_2$, $C_4$, or $C_5$ in $G''$.
  Since neither the transformation from $G$ to $G'$ nor the transformation from $G'$ to $G''$ makes any change to $V(G)\setminus I_{0}$, this set induces the same subgraph in $G$ and $G''$.  Thus, $F\not\subseteq V(G)\setminus I_{0}$.  Moreover, since every vertex in $K$ is universal in $G'' - I_{0}'$, it follows that $F\cap I_{0}'$ is not empty.  Note that every vertex in $I_{0}'$ has only one neighbor in $G''$, we can conclude that $G''[F]$ must be a $2 K_{2}$ and $F\cap K = \emptyset$.  But then $v$ is contained in either a $2 K_{2}$ or an $I_{0}$-centered $P_{3}$ in $G$, a contradiction.

  It then follows from Lemma~\ref{lem:safe-vertex-split} that $(G'', k)$ is equivalent to $(G''- v, k)$.  To see the equivalence between $(G''- v, k)$ and $(G-v, I_{0}, k)$, we apply the reversed operations from $G$ to $G''$.  We first use Rule~\ref{rule:decided-vertices}, applied to $(G''- v, \emptyset, k)$, to mark all vertices in $I_{0}'$, then use Lemma~\ref{lem:i-vertices} to replace $I_{0}'$ by $I_{0}$, and finally remove vertices in $K$.  The resulting graph is precisely $G - v$.  We can thus conclude the proof.
\end{proof}

We call an annotated instance \emph{reduced} if none of Rules~\ref{rule:simplicial}--\ref{rule:alternative} is applicable to this instance.
The following lemma bounds the cardinalities of $C_{M}$ and $I_{M}$ in a reduced instance.
\begin{lemma} \label{lemma:bound}
  If a reduced instance $(G, I_{0}, k)$ is a yes-instance, then $|C_{M}| \le 3k\sqrt{2 k}$ and $|I_{M}|\le k+1$.
\end{lemma}
\begin{proof}
  Let $E_{-}$ be any solution to $(G, I_{0}, k)$ with at most $k$ edges.
  Since Rule~\ref{rule:alternative} is not applicable, every vertex in $C_{M}$ is contained in some $2K_2$ or $I_{0}$-centered $P_{3}$, or some $C_4$ or $C_5$ in $G - I_{0}$.  Any of these structures contains an edge in $E_{-}$.  Therefore, to bound $|C_{M}|$, it suffices to count how many vertices in $C_{M}$ can form a $2K_2$ or $I_{0}$-centered $P_{3}$, or a $C_4$ or $C_5$ in $G - I_{0}$ with an edge $xy\in E_{-}$.
  \begin{itemize}
  \item  If a vertex $v\in C_{M}$ is in a $2 K_{2}$ with edge $x y$, then either $v\in \{x, y\}$ or $v$ is adjacent to neither $x$ nor $y$.  In the first case, no other vertex in $C_{M}$ can occur in any $2 K_2$ with $x y$.
  Since $x y\in E(G)$, at least one of them is not in $I_{0}$ (Rule~\ref{rule:i-edges}).  This vertex has at most $\sqrt{2 k}$ non-neighbors in $C_{M}$.  Therefore, the total number of vertices in $C_{M}$ that can occur in any $2 K_2$ with $x y$ is at most $\sqrt{2 k}$.
\item If $x y$ is an edge in any $I_{0}$-centered $P_{3}$, then precisely one of them is in $I_{0}$.  Assume without loss of generality $x\in I_{0}$.  If a vertex $v\in C_{M}$ is in an $I_{0}$-centered $P_{3}$ with the edge $xy$, then either $v = y$, or $v$ is not adjacent to $y$.  Since $y\not\in I_{0}$, it has at most $\sqrt{2 k}$ non-neighbors in $C_{M}$.  Thus, the total number of vertices in $C_{M}$ that can occur in any $I_{0}$-centered $P_{3}$ containing $x y$ is at most $\sqrt{2 k} + 1$.
\item  If a vertex $v\in C_{M}$ is in a $C_{4}$ or $C_{5}$ that contains $x y$, then $v$ is adjacent to at most one of $x$ and $y$.  Since this $C_{4}$ or $C_{5}$ is in $G - I_{0}$, each of $x$ and $y$ has at most $\sqrt{2 k}$ non-neighbors in $C_{M}$.  Thus, the total number of vertices in $C_{M}$ that can occur in such a $C_4$ or $C_5$ is at most $2 \sqrt{2 k}$.
  \end{itemize}
Noting that an edge cannot satisfy the conditions of both the second ($|\{x, y\}\cap I_{0}| = 1$) and third ($|\{x, y\}\cap I_{0}| = 0$) categories, we can conclude $|C_{M}| \le k(\sqrt{2 k} + 2 \sqrt{2 k}) = 3 k \sqrt{2 k}$.

  Since Rule~\ref{rule:simplicial} is not applicable, no vertex in $I_{M}$ is simplicial.
  Suppose that $C\uplus I$ is a valid partition of $G$.
  Since $C$ is a clique, for each vertex $v\in I_{M}\cap I$, at least one neighbor of $v$ is in $I$.  Therefore, each vertex $v\in I_{M}\cap I$ is incident to an edge in the solution $E(I, I)$.  Noting that $I_{M}$ is an independent set, we have
  $k \ge |I_{M}\cap I| \ge |I_{M}| - 1$, where the second inequality follows from Lemma~\ref{lem:partition}(i).  Thus, $|I_{M}| \le k + 1$, and this concludes this proof.
\end{proof}

Note that the application of Rule~\ref{rule:clean-up} is different from the other ones.
The application of one of Rules~\ref{rule:simplicial}--\ref{rule:alternative} may trigger the applicable of another.
After the application of Rule~\ref{rule:clean-up}, the instance is no longer annotated, and we will not go back to check the other rules.
We summarize the algorithm in Figure~\ref{fig:alg-split}.

\begin{figure}[h!]
  \centering
  \begin{tikzpicture}
    \path (0,0) node[text width=.73\textwidth, inner xsep=20pt, inner ysep=10pt] (a) {
      \begin{minipage}[t!]{\textwidth} \small
        {\sc Input}: an instance $(G, k)$ of the split edge deletion problem.
        \\
        {\sc Output}: an equivalent instance $(G', k')$ with $|V(G')| = O(k'^{1.5})$.
        
        \begin{tabbing}
          AAA\=Aaa\=aaa\=Aaa\=MMMMMMAAAAAAAAAAAAA\=A \kill
          
          1.\> $I_{0}\leftarrow \emptyset$;
          \\
          1.\> $M\leftarrow$ a maximal packing of vertex-disjoint $2 K_{2}$'s, $C_{4}$'s, and $C_{5}$'s;
          \\
          2.\> {\bf if} $|M| > 5 k$ {\bf then return} a trivial no-instance;
          \\
          3. \> {\bf if} $k < 0$ {\bf then return} a trivial no-instance;
          \\
          4.\> \textbf{for each} simplicial vertex $v\in V(G)\setminus I_{0}$ {\bf do} (Rule~\ref{rule:simplicial})
          \\
          \>\> {\bf if} $|E(G - (N[v]\setminus I_{0}) )| \le k$ {\bf then return} a trivial yes-instance;
          \\
          \>\> {\bf else}  $I_{0}\leftarrow I_{0}\cup \{v\}$;
          \\
          5.\> remove c-vertices and mark i-vertices (Rule~\ref{rule:decided-vertices});
          \\
          6.\> {\bf if} $E(I_{0}, I_{0})\ne \emptyset$ {\bf then}
          \\
          \>\> remove edges in $E(I_{0}, I_{0})$ and decrease $k$ (Rule~\ref{rule:i-edges});
          \\
          7.\> merge $I_{0}$ into $\le k$ vertices (Rule~\ref{rule:merge-I0});
          \\
          8.\> remove vertices in $C_{M}$ not contained in certain structures (Rule~\ref{rule:alternative});
          \\
          9.\> {\bf if} any of Rules~\ref{rule:simplicial}--\ref{rule:i-edges} and \ref{rule:alternative} made a change {\bf then} \textbf{goto} 3;
          \\
          10.\> {\bf if} $|C_{M}| + |I_{M}|> 3 k \sqrt{2 k} + k + 1$ {\bf then return} a trivial no-instance;
          \\ 
          11.\> add $\sqrt{2 k} + 1$ new vertices and remove all marks (Rule~\ref{rule:clean-up});
          \\
          12.\> {\bf return} $(G, k)$.
        \end{tabbing}
      \end{minipage}
    };
    \draw[draw=gray!60] (a.north west) -- (a.north east) (a.south west) -- (a.south east);
  \end{tikzpicture}
  \caption{A summary of our kernelization algorithm for split edge deletion.
}
  \label{fig:alg-split}
\end{figure}

\begin{theorem}\label{thm:split}
  There is an $O(k^{1.5})$-vertex kernel for the split edge deletion problem.  
\end{theorem}
\begin{proof}
  We use the algorithm described in Figure~\ref{fig:alg-split}.
  The first two steps build the modulator, and their correctness follows from that any solution contains at least one edge of each forbidden induced subgraph of $G$.  Step~3 is obviously correct.  Steps~4--8 follow from the safeness of the rules; so is step~11.  The correctness of step~10 is ensured by Lemma~\ref{lemma:bound}.
  
  The cardinality of $M$ is at most $5 k$, and it never increases during the algorithm.  After step~7, $|I_{0}| \le k$.  We have bounded the cardinalities of $C_{M}$ and $I_{M}$ in Lemma~\ref{lemma:bound}.  Step~11 increases $|C_{M}|$ by $\sqrt{2 k} + 1$.  Putting them together, we have
  \[
    |V(G)| \le
    5 k + k + (3 k \sqrt{2 k} + \sqrt{2 k} + 1) + k + 1 = O(k^{1.5}).
  \]

  It is easy to verify that each reduction rule can be checked and applied in polynomial time.
  To see that the algorithm runs in polynomial time, note that if any of Rules~\ref{rule:simplicial}--\ref{rule:i-edges} and \ref{rule:alternative} made a change to the instance, then either $k$ decreases by one (Rule~\ref{rule:i-edges}), or the cardinality of $V(G)\setminus I_{0}$ decreases by one (the other three rules).
\end{proof}

Since the class of split graphs is self-complementary, our algorithm also implies a kernel for the split completion problem.  This kernel actually has fewer edges than the one for split edge deletion.

\begin{theorem}
    There is a kernel of $O(k^{1.5})$ vertices and $O(k^{2.5})$ edges for the split completion problem.  
\end{theorem}
\begin{proof}
  Let $(G, k)$ be the input instance of the split completion problem.  
  We can either take the complement of the input graph and consider it as an instance of the split edge deletion problem, or run the ``complemented versions'' of the rules.  In the final result, we have an independent set of at most $O(k\sqrt{k})$ vertices, and at most $O(k)$ other vertices.  The claim then follows.
\end{proof}

\section{Pseudo-split edge deletion and pseudo-split completion}

A \emph{pseudo-split graph} is either a split graph, or a graph whose vertex set can be partitioned into a clique $C$, an independent set $I$, and a set $S$ such that (1) $S$ induces a $C_{5}$; (2) $C\subseteq N(v)$ for every $v\in S$; and (3) $I\cap N(v) = \emptyset$ for every $v\in S$.
We say that $C\uplus I \uplus S$ is a \emph{pseudo-split partition} of the graph, where $S$ may or may not be empty.
If $S$ is empty, then $C\uplus I$ is a split partition of the graph.  Otherwise, the graph has a unique pseudo-split partition.  (One may also verify that $S$ is a module.)
The forbidden induced subgraphs of pseudo-split graphs are $2K_2$ and $C_4$ \cite{maffray-94-pseudo-split}.
Similar as split graphs, the complement of a pseudo-split graph remains a pseudo-split graph.  Thus, the completion problem and the edge deletion problem toward pseudo-split graphs are polynomially equivalent.  In this section, we study the pseudo-split edge deletion problem.

The class of pseudo-split graphs is a superclass of split graphs.  In particular, split graphs are precisely $C_{5}$-free pseudo-split graphs.  Note that a pseudo-split graph contains at most one $C_{5}$.  In case that a pseudo-split graph does contain a $C_{5}$, removing any vertex from the $C_{5}$ leaves a split graph.  Therefore, those two classes are very ``close.''
Another way to derive a split subgraph from a pseudo-split graph is by removing any two consecutive edges from the $C_{5}$. 
Thus, if we use $\mathrm{sed}(G)$ to denote the size of the smallest edge set $E_{-}$ such that $G - E_{-}$ is a split graph, then
\[
  \mathrm{opt}(G) \le \mathrm{sed}(G) \le \mathrm{opt}(G) + 2.
\]
Moreover, $\mathrm{opt}(G) = \mathrm{sed}(G)$
if and only if there is a minimum solution $E_{-}$ of $G$ (for the pseudo-split edge deletion problem) such that $G - E_{-}$ is a split graph.

We say that a partition $C\uplus I \uplus S$ of the vertex set of the input graph $G$ is a \emph{valid partition} of the instance $(G, k)$ if there exists a set $E_{-}$ of at most $k$ edges such that $C\uplus I \uplus S$ is a pseudo-split partition of $G - E_{-}$.

\begin{proposition}\label{lem:induced-c5}
  Let $G$ be a graph with $\mathrm{opt}(G) < \mathrm{sed}(G)$, let $E_{-}$ be a minimum solution to $G$, and let $C\uplus I\uplus S$ be the pseudo-split partition of $G - E_{-}$.
  Then $G[S]$ is a $C_{5}$, and no vertex in $I$ forms a triangle with two vertices in $S$ in $G$.
\end{proposition}
\begin{proof}
  Since $\mathrm{opt}(G) < \mathrm{sed}(G)$, the set $S$ cannot be empty.
  Let $G' = G - E_{-}$, and let $v_1 v_2 v_3 v_4v_5$ be the cycle of $G'[S]$.  We first argue that $G[S]$ is a $C_{5}$.  Suppose otherwise, then there is a chord of the cycle $G'[S]$, say $v_1 v_3$, in $E_{-}$.  
  We take $E'_{-} = (E_{-}\setminus \{v_1 v_3\})\cup \{v_4 v_5\}$.
  Note that $G - E'_{-}$ is a split graph, as evidenced by the split partition $(C\cup \{v_1, v_2, v_3\})\uplus (I\cup \{v_4, v_5\})$.  But $|E'_{-}| = |E_{-}|$ contradicts $\mathrm{opt}(G) < \mathrm{sed}(G)$.

  For the second part, suppose for contradiction that there is a triangle of $G$ containing a vertex $u\in I$ and two vertices in $S$.  We have seen that $G[S]$ is a $C_{5}$.  We may assume that the triangle is $u v_{1} v_{2}$.  We take $E'_{-} = (E_{-}\setminus \{u v_1, u v_2\})\cup \{v_3 v_4, v_4 v_5\}$.
  Note that $G - E'_{-}$ is a split graph, as evidenced by the split partition $(C\cup \{v_1, v_2\})\uplus (I\cup \{v_3, v_4, v_5\})$.  But $|E'_{-}| = |E_{-}|$ contradicts $\mathrm{opt}(G) < \mathrm{sed}(G)$.
\end{proof}

For the pseudo-split edge deletion problem, one may expect a proposition similar as Lemmas~\ref{lem:stable-vertex} and~\ref{lem:safe-vertex-split}; i.e., it is safe to remove vertices not in any $2K_2$ or $C_4$.  As shown in Figure~\ref{fig:pseudo-split-example}, this is however not true.  This graph contains no $2 K_{2}$, and the only $C_{4}$ is $v_{1} v_{2} v_{3} v_{4}$.  
The deletion of any edge from this cycle introduces a new $2 K_{2}$, e.g., $\{v_{1} v_{4}, v_{2} v_{6}\}$ after $v_{1} v_{2}$ deleted.
On the other hand, $\mathrm{opt}(G - v_{6}) = 1$ because it suffices to delete either $v_{1} v_{2}$ or $v_{2} v_{3}$.
We manage to show that if a vertex $v$ is not in any $2K_2$, $C_4$, or $C_5$, then it is safe to remove $v$.  This is sufficient for our algorithm.

\begin{figure}[h]
  \centering
  \begin{tikzpicture}[scale=1.2]
    \node["$v_1$" below, filled vertex] (v1) at (90:.3) {};
    \foreach \i in {1,..., 5} {
      \draw ({\i * (360 / 5) - 54}:1) -- ({\i * (360 / 5) + 18}:1);
    }
    \foreach \i in {2,..., 6} {
      \node[filled vertex] (v\i) at ({\i * (360 / 5) - 126}:1) {};
      \node at ({\i * (360 / 5) - 126}:1.25) {$v_{\i}$};
    }
    \foreach \i in {2, 4} \draw (v\i) -- (v1);
  \end{tikzpicture}
  \caption{
    $\mathrm{opt}(G) = 2$, while $\mathrm{opt}(G - v_{6}) = 1$.}
  \label{fig:pseudo-split-example}
\end{figure}

\begin{lemma}\label{lem:safe-vertex-pseudo-split}
  If a vertex $v$ is not contained in any $2K_2$, $C_4$, or $C_5$, then $\mathrm{opt}(G - v) = \mathrm{opt}(G)$.
\end{lemma}
\begin{proof}
  Let $G' = G - v$.  It is trivial that any solution to $G$ contains a solution to $G'$, hence $\mathrm{opt}(G') \le \mathrm{opt}(G)$.
  Let $E_{-}$ be a minimum solution to $G'$.  We have nothing to show if $G - E_{-}$ is a pseudo-split graph as well.  In the rest of the proof, $G - E_{-}$ is not a pseudo-split graph.
  We take a pseudo-split partition $C\uplus I\uplus S$ of $G'-E_{-}$. 
  If $S=\emptyset$, then $G' - E_{-}$ is a split graph.   Since every split graph is a pseudo-split graph, $E_{-}$ is also a minimum split edge deletion set of $G - v$.  By Lemma~\ref{lem:safe-vertex-split}, there is a set $E_{-}'$ such that $|E_{-}'| = |E_{-}|$ and $G - E_{-}'$ is a split graph.  Thus, $\mathrm{opt}(G) \le |E_{-}'| = \mathrm{opt}(G - v)$.
  
  Now that $S\neq \emptyset$, we may assume without loss of generality that (1) $G[S]$ is a $C_{5}$, and (2) $|S\cap N_{G}(x)| \le 2$ for every $x\in I$; otherwise, by Proposition~\ref{lem:induced-c5}, we can find another solution $E_{-}'$ of $G'$ such that  $|E_{-}'| = |E_{-}|$ and $G' - E_{-}'$ is a split graph, and then we are in the previous case.
  Under these assumptions we show that $v$ is either adjacent to all vertices in $C\cup S$, or nonadjacent to any vertex in $I\cup S$.  Accordingly, either 
  $(C\cup \{v\})\uplus I\uplus S$ or $C\uplus(I\cup \{v\})\uplus S$ is a
   pseudo-split partition of $G-E_{-}$, and hence $G-E_{-}$ is also a pseudo-split graph.
  
   Let us start from the adjacency between $v$ and $S$.
   If $v$ is adjacent to only one vertex in $S$, or two or three consecutive vertices on the $C_{5}$, then $v$ is contained in a $2 K_{2}$.  On the other hand, if $v$ is adjacent to four vertices in $S$, or two or three non-consecutive vertices on the $C_{5}$, then $v$ is contained in a $C_{4}$.  See Figure~\ref{fig:c5} for illustration.
   Therefore, $v$ is adjacent to either all or none of the vertices in $S$.
   If $S\subseteq N_G(v)$, then $C\subseteq N_G(v)$ as well: $v$, a vertex $x\in C\setminus N_G(v)$, and two nonadjacent vertices in $S$ would induce a $C_{4}$.
   Now that $N_G(v)\cap S=\emptyset$, we are done if $N_G(v)\cap I=\emptyset$.  Suppose otherwise, and let $u$ be any vertex in $N_G(v)\cap I$.
  By assumption (2), $|S\cap N_{G}(u)| \le 2$.  But then an edge in $G[S]$ of which both ends nonadjacent to $u$ form an induced $2K_2$ with $uv$ in $G$.
  This concludes the proof.
\end{proof}

\begin{figure}[h]
  \centering\small
  \begin{subfigure}[b]{.15\linewidth}
    \centering
    \begin{tikzpicture}[scale=1.]
      \node["$v$" below, filled vertex] (v0) at (0, 0) {};
      \foreach \i in {1,..., 5} {
        \draw ({\i * (360 / 5) - 54}:1) -- ({\i * (360 / 5) + 18}:1);
      }
      \foreach \i in {1,..., 5} {
        \node[filled vertex] (v\i) at ({\i * (360 / 5) - 54}:1) {};
      }
      \draw[very thick, RubineRed] (v2) -- (v0) (v4) -- (v5);
    \end{tikzpicture}
    \caption{}
  \end{subfigure}  
  \,
  \begin{subfigure}[b]{.15\linewidth}
    \centering
    \begin{tikzpicture}[scale=1.]
      \node["$v$" below, filled vertex] (v0) at (0, 0) {};
      \foreach \i in {1,..., 5} {
        \draw ({\i * (360 / 5) - 54}:1) -- ({\i * (360 / 5) + 18}:1);
      }
      \foreach \i in {1,..., 5} {
        \node[filled vertex] (v\i) at ({\i * (360 / 5) - 54}:1) {};
      }
      \foreach \i in {4, 5} \draw (v\i) -- (v0);
      \draw[very thick, RubineRed] (v2) -- (v1) (v4) -- (v0);
    \end{tikzpicture}
    \caption{}
  \end{subfigure}  
  \,
  \begin{subfigure}[b]{.15\linewidth}
    \centering
    \begin{tikzpicture}[scale=1.]
      \node["$v$" below, filled vertex] (v0) at (0, 0) {};
      \foreach \i in {1,..., 5} {
        \draw ({\i * (360 / 5) - 54}:1) -- ({\i * (360 / 5) + 18}:1);
      }
      \foreach \i in {1,..., 5} {
        \node[filled vertex] (v\i) at ({\i * (360 / 5) - 54}:1) {};
      }
      \foreach \i in {1, 2, 3} \draw (v\i) -- (v0);
      \draw[very thick, RubineRed] (v2) -- (v0) (v4) -- (v5);
    \end{tikzpicture}
    \caption{}
  \end{subfigure}  
  \,
  \begin{subfigure}[b]{.15\linewidth}
    \centering
    \begin{tikzpicture}[scale=1.]
      \node["$v$" below, filled vertex] (v0) at (0, 0) {};
      \foreach \i in {1,..., 5} {
        \draw ({\i * (360 / 5) - 54}:1) -- ({\i * (360 / 5) + 18}:1);
      }
      \foreach \i in {1,..., 5} {
        \node[filled vertex] (v\i) at ({\i * (360 / 5) - 54}:1) {};
      }
      \foreach \i in {1, 3} \draw (v\i) -- (v0);
      \draw[very thick, RubineRed] (v0) -- (v1) -- (v2) -- (v3) -- (v0);
    \end{tikzpicture}
    \caption{}
  \end{subfigure}  
  \,
  \begin{subfigure}[b]{.15\linewidth}
    \centering
    \begin{tikzpicture}[scale=1.]
      \node["$v$" below, filled vertex] (v0) at (0, 0) {};
      \foreach \i in {1,..., 5} {
        \draw ({\i * (360 / 5) - 54}:1) -- ({\i * (360 / 5) + 18}:1);
      }
      \foreach \i in {1,..., 5} {
        \node[filled vertex] (v\i) at ({\i * (360 / 5) - 54}:1) {};
      }
      \foreach \i in {2, 4, 5} \draw (v\i) -- (v0);
      \draw[very thick, RubineRed] (v0) -- (v2) -- (v1) -- (v5) -- (v0);
    \end{tikzpicture}
    \caption{}
  \end{subfigure}  
  \,
  \begin{subfigure}[b]{.15\linewidth}
    \centering
    \begin{tikzpicture}[scale=1.]
      \node["$v$" below, filled vertex] (v0) at (0, 0) {};
      \foreach \i in {1,..., 5} {
        \draw ({\i * (360 / 5) - 54}:1) -- ({\i * (360 / 5) + 18}:1);
      }
      \foreach \i in {1,..., 5} {
        \node[filled vertex] (v\i) at ({\i * (360 / 5) - 54}:1) {};
      }
      \foreach \i in {1, 3, 4, 5} \draw (v\i) -- (v0);
      \draw[very thick, RubineRed] (v0) -- (v1) -- (v2) -- (v3) -- (v0);
    \end{tikzpicture}
    \caption{}
  \end{subfigure}  
  \caption{Illustration for the proof of Lemma~\ref{lem:safe-vertex-pseudo-split}.  There is a $2 K_{2}$ (a, b, c) or $C_{4}$ (d, e, f), shown as thick lines, containing $v$.}
  \label{fig:c5}
\end{figure}
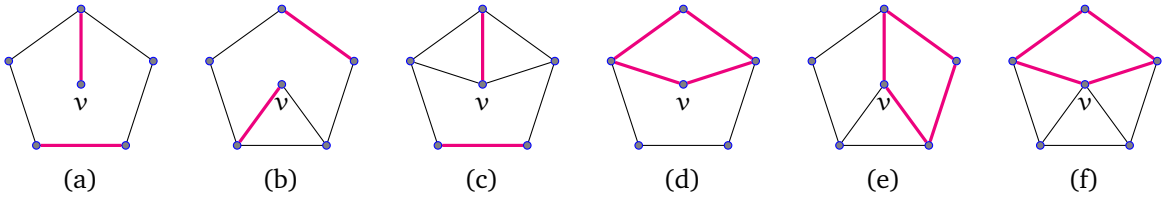

To adapt the algorithm in Figure~\ref{fig:alg-split} for the pseudo-split edge deletion problem, we only need to conduct very minor adjustments.  We use the same modulator $M$ as the previous section, i.e., vertices of a maximal packing of vertex-disjoint $2 K_{2}$'s, $C_{4}$'s, and $C_{5}$'s.  Recall that a pseudo-split graph contains at most one $C_{5}$.  Thus, if we have found $p$ vertex-disjoint $C_{5}$'s from $G$, then we need to break at least $p - 1$ of them.

\begin{lemma}
  If $(G, k)$ is a yes-instance, then $|M| \le 4 k + 5$.
\end{lemma}
\begin{proof}
  Let $E_{-}$ be a minimum solution to $G$.
  Suppose that the numbers of vertex-disjoint $2 K_{2}$'s, $C_{4}$'s, and $C_{5}$'s we have put into $M$ are $p$, $q$, and $r$, respectively.  In each $2 K_{2}$ or $C_{4}$, at least one edge needs to be in $E_{-}$.  At most one $C_{5}$ can be disjoint from $E_{-}$, and on each of other $C_{5}$, at least two edges are in $E_{-}$.  Thus
  $k\ge |E_{-}| \ge p + q + 2 (r - 1)$, and $M = 4 p + 4 q + 5 r = 4 p + 4 q + 5 (r - 1) + 5 \le 4 k + 5$.
\end{proof}

As a result, if $|M| > 4 k + 5$, then $(G, k)$ must be a no-instance (step~2 of the algorithm).
Again, we start from a split partition $C_{M}\uplus I_{M}$ of $G - M$, and we work on the annotated version of the problem.
In an annotated instance $(G, I_{0}, k)$, the set $I_{0}$ of marked vertices can only be put into the independent set $I$ in a valid partition.
We use the same rules as we have used for the split edge deletion problem.  We now verify that all of them remain safe for the pseudo-split edge deletion problem.
The first is simple.
\begin{reduction}\label{rule:clean-up(pseudo-split)}
  Let $(G, I_{0}, k)$ be an annotated instance.
  Add a clique of $\sqrt{2 k} + 1$ new vertices, and make each of them adjacent to all the vertices in $V(G)\setminus I_{0}$.  Return as an unannotated instance.
\end{reduction}
\begin{proof}[Safeness of Rule~\ref{rule:clean-up(pseudo-split)}.]
  Let $K$ denote the clique of new vertices, and let $(G', k)$ be the resulting instance.
  For any valid partition $C\uplus I\uplus S$ of $(G, I_{0}, k)$, the partition $(C\cup K)\uplus I\uplus S$ is a valid partition of $(G', k)$ because $C\subseteq V(G)\setminus I \subseteq N(x)$ and $S\subseteq V(G)\setminus I\subseteq N(x)$ for every $x\in K$.
  Let $E_-$ be a set of at most $k$ edges such that $G'-E_-$ is a pseudo-split graph and $C\uplus I\uplus S$ is a pseudo-split partition of $G' - E_{-}$.
  If any vertex in $I_{0}$ is in $C\cup S$, then we must have $K\subseteq I$.  
  But then $|E_-| > k$, which contradicts the validity of the partition.
\end{proof}

Rule~\ref{rule:simplicial} was safe for the split edge deletion problem because a vertex is either in $C$ or $I$.  For the pseudo-split edge deletion problem, we need to take care of the possibility that a simplicial vertex is in $S$.

\begin{lemma}\label{lem:no simplicial in S}
  Let $(G, I_{0}, k)$ be an instance of the annotated version of the pseudo-split edge deletion problem, and let $v$ be a simplicial vertex of $G$.
  If there exists a valid partition $C\uplus I \uplus S$ with $v\in S$, then there exists another valid partition $C'\uplus I' \uplus S'$ with $S' = \emptyset$. 
\end{lemma}
\begin{proof}
  Since $v\in S$, the set $S$ is not empty.  Since $v$ is simplicial, $G[S]$ is not a $C_{5}$.  The statement follows from Proposition~\ref{lem:induced-c5}.
\end{proof}

As a result, for any simplicial vertex $v$, it suffices to look for a valid partition $C\uplus I \uplus S$ with $v\in C\cup I$.  If $v\in C$, then $S = \emptyset$; since $C\cup S$ is a subset of $N[v]$, it is a clique, and then $(C\cup S)\uplus I$ is a partition with a smaller solution.  Thus, it remains safe.

\begin{reduction}\label{rule:simplicial(pseudo-split)}
  Let $v$ be a simplicial vertex in $V(G)\setminus I_{0}$. 
  If $|E(G - (N[v]\setminus I_{0}) )| \le k$, then return a trivial yes-instance.  
  Otherwise, add $v$ to $I_{0}$.
\end{reduction}
\begin{proof}[Safeness of Rule~\ref{rule:simplicial(pseudo-split)}.]
  In the first case, $(N[v]\setminus I_{0}) \uplus (V(G)\setminus N[v] \cup I_{0})$ is a valid partition.
  Otherwise, we show that there is a minimum solution $E_{-}$ to $(G,I_0,k)$ such that $C\uplus I\uplus S$ is a pseudo-split partition of $G - E_{-}$ and $v\in I$. 
  Suppose for contrary that $v\notin I$, then by Lemma~\ref{lem:no simplicial in S}, we can suppose that $v$ is in $C$. 
  Since $C$ is a clique, then $C\subseteq N[v]\setminus I_{0}$. 
  Since every vertex in $S$ is adjacent to all vertices in $C$, then $S\subseteq N[v]\setminus I_0$.  
  Thus, $E(G - (N[v]\setminus I_{0}))\subseteq E_-$, but then $|E_-| > k$, contradicting the validity of the partition.
\end{proof}

The safeness of Rules~\ref{rule:decided-vertices}, \ref{rule:i-edges}, \ref{rule:merge-I0}, and \ref{rule:alternative} can be proved with almost the same arguments as in Section~\ref{sec:split}.

\begin{lemma}\label{lem:pseudo-split partition}
  Let $C\uplus I\uplus S$ be a valid partition of $(G,k)$, if one exists, 
    \begin{enumerate}[i)]
    \item $|I_{M}\cap C| \le 1$;
    \item $|I_M\cap S|\leq 2$; and 
    \item $|C_{M}\cap I| \le \sqrt{2 k}$.
    \end{enumerate}
  \end{lemma}
  \begin{proof}
  The first assertion follows from that $I_{M}$ is an independent set and $C$ is a clique. 
  The second assertion holds because a $C_{5}$ does not contain an independent set of order three.
  The last assertion holds because 
  \[
    k \ge |E(I, I)| \ge |E(C_{M}\cap I, C_{M}\cap I)| = {|C_{M}\cap I|\choose 2}. \qedhere
  \]
  \end{proof}

  We say that a vertex is a \emph{c-vertex}, respectively, an \emph{i-vertex}, if it is in $C$, respectively, in $I$, for any valid partition $C\uplus I\uplus S$ of $(G, I_{0}, k)$ (we do not consider vertices in $S$ since there are only five such vertices).
  \begin{lemma}
    If a vertex $v$ has more than $k+3$ neighbors in $I_M$, then $v$ is a c-vertex. 
      \label{lemma:c-vertex(pseudo-split)}
  \end{lemma}
  \begin{proof}
    Let $E_-$ be a set of at most $k$ edges such that $G-E_-$ is a pseudo-split graph and $C\uplus I\uplus S$ is a pseudo-split partition of $G-E_-$.
      Suppose for contrary that $v$ is not a c-vertex, then it is in $I$ or $S$.
      By Lemma~\ref{lem:pseudo-split partition} (i) and (ii), there are at least $k+1$ neighbors of $v$ are in $I$, then $|E_-|>k$, which is a contradiction.     
  \end{proof}
\begin{lemma}
    The following are i-vertices:
    \begin{itemize}
      \item every vertex nonadjacent to a c-vertex; and
      \item every vertex with more than $\sqrt{2k}+1$ non-neighbors in $C_{M}$.
      \end{itemize}
    \label{lemma:i-vertex(pseudo-split)}
\end{lemma}
\begin{proof}
    The first assertion follows from the definition of the pseudo-split graphs.
    Now we proof the second assertion. 
    Let $v$ be a vertex with more than $\sqrt{2k} + 1$ non-neighbors in $C_M$, and let $E_-$ be a set of at most $k$ edges such that $G-E_-$ is a pseudo-split graph and $C\uplus I\uplus S$ is a pseudo-split partition of $G - E_{-}$.
    Suppose for contrary that $v\notin I$, then $v$ is in $C$ or $S$.
    If $v\in C$, then all non-neighbors of $v$ in $C_M$ are in $I$, then $|C_M\cap I|>\sqrt{2k}$, which contradicts Lemma~\ref{lem:pseudo-split partition} (iii).
    If $v\in S$, then there are at most two non-neighbors of $v$ in $C_M$ that are in $S$, then all but one edges in $G[(V(G)\setminus N[v])\cap C_M]$ are in $E_-$, then $|E_-|>k$, a contradiction.  
\end{proof}

We can indeed delete all the c-vertices, as long as we keep their non-neighbors marked. 
Note that after obtaining the initial split partition $C_{M}\uplus I_{M}$ of $G - M$, we do not need to maintain the invariant that $M$ is a modulator, though we do maintain that $C_{M}$ is a clique and that $I_{M}$ is an independent set throughout.
During our algorithm, we maintain $M$, $C_{M}$, $I_{M}$, and $I_{0}$ as a partition of $V(G)$.
 Therefore, whenever we mark a vertex, we move it to $I_{0}$.
\begin{reduction}\label{rule:decided-vertices(pseudo-split)}
  Let $(G, I_{0}, k)$ be an annotated instance.
  \begin{enumerate}[i)]
  \item Mark every vertex that has more than $\sqrt{2 k}+1$ non-neighbors in $C_{M}$.
  \item If a vertex $v$ has more than $k+3$ neighbors in $I_{M}\cup I_{0}$, then mark every vertex in $V(G)\setminus N[v]$ and delete $v$.
  \end{enumerate}
\end{reduction}
\begin{proof}[Safeness of Rule~\ref{rule:decided-vertices(pseudo-split)}.]
  Let $I_{0}'$ denote the set of marked vertices after the reduction.  It is trivial that if the resulting instance of i) is a yes-instance, then the original is also a yes-instance.  For ii), any valid partition $C'\uplus I'\uplus S'$ of $(G - v, I_{0}', k)$ can be extended to a valid partition $(C'\cup \{v\})\uplus I'\uplus S'$ of $(G, I_{0}, k)$ because $C'\subseteq V(G)\setminus I_{0}'\subseteq N[v]$ and $S'\subseteq V(G)\setminus I_{0}'\subseteq N[v]$.
  
  For the other direction, let $E_-$ be a set of at most $k$ edges such that $G-E_-$ is a pseudo-split graph and $C\uplus I\uplus S$ is a pseudo-split partition of $G-E_-$ where $I_0\subseteq I$. 
  i) Since $|E_-|\leq k$, then $|C_{M}\setminus N(v)| \le \sqrt{2 k}$ for every $v\in C$ and $|C_{M}\setminus N(v)| \le \sqrt{2 k}+1$ for every $v\in S$.  
  Therefore, if $|C_{M}\setminus N(v)| > \sqrt{2 k}+1$ for some vertex $v$, then $v$ has to be in $I$. 
  Thus, $C\uplus I\uplus S$ is also a valid partition of the new instance $(G, I_{0}', k)$.
  ii) By Lemma~\ref{lem:pseudo-split partition}(i) and (ii), $|I_{M}\cap I|\le 1$ and $|I_M\cap S|\le 2$. 
  As $I_{0}\subseteq I$ and $|N(v)\cap (I_{M}\cup I_{0})| > k + 3$, there are at least $k + 1$ edges between $v$ and $I$. 
  Since $|E_-| \le k$, we must have $v\in C$.  Moreover, since $C$ is a clique, $C\subseteq N[v]$, and every vertex nonadjacent to $v$ has to be in $I$.  This justifies the marking of $V(G)\setminus N[v]$.  Clearly, $(C\setminus \{v\})\uplus I\uplus S$ is a valid partition of $(G - v, I_{0}', k)$.
\end{proof}

\begin{reduction}\label{rule:i-edges(pseudo-split)}
  Let $(G, I_{0}, k)$ be an annotated instance.
  Remove all the edges in $E(I_{0}, I_{0})$, and decrease $k$ accordingly.    
\end{reduction}
\begin{proof}[Safeness of Rule~\ref{rule:i-edges(pseudo-split)}.]
By the definition of the annotated instance, any solution $E_{-}$ of $(G, I_{0}, k)$ contains all the edges in $E(I_{0}, I_{0})$.  Moreover, $E_{-}\setminus E(I_{0}, I_{0})$ is a solution to $G - E(I_{0}, I_{0})$, and its size is at most $k - |E(I_{0}, I_{0})|$.  On the other hand, if $(G - E(I_{0}, I_{0}), k - |E(I_{0}, I_{0})|)$ is a yes-instance, then any solution of this instance, together with $E(I_{0}, I_{0})$, makes a solution of $(G, I_{0}, k)$ of size at most $k$.
\end{proof}

\begin{reduction}\label{rule:merge-I0(pseudo-split)}
  Let $(G, I_{0}, k)$ be an annotated instance where $I_{0}$ is an independent set.
  Introduce $p$ new vertices $v_{1}$, $v_{2}$, $\ldots$, $v_{p}$, where $p = \max_{v\in V(G)\setminus I_0}|N(v)\cap I_{0}|$.  For each vertex $x\in N(I_{0})$, make $x$ adjacent to $v_{1}$, $\ldots$, $v_{|N(x)\cap I_{0}|}$.  Remove all vertices in $I_{0}$, and mark the set of new vertices.
\end{reduction}

The following statement ensures the safeness of Rule~\ref{rule:merge-I0(pseudo-split)}.
Note that if Rule~\ref{rule:decided-vertices(pseudo-split)} is not applicable, then $p \le k+2$.  
\begin{lemma}\label{lem:i-vertices(pseudo-split)} 
  Let $(G, I_{0}, k)$ and $(G', I_{0}', k)$ be two annotated instances where $G - I_{0} = G' - I_{0}'$ and both $I_{0}$ and $I_{0}'$ are independent sets.  If $|N_{G}(x)\cap I_{0}| = |N_{G'}(x)\cap I_{0}'|$ for every $x\in V(G)\setminus I_{0}$, then $(G, I_{0}, k)$ is a yes-instance if and only if $(G', I_{0}', k)$ is a yes-instance.
\end{lemma}
\begin{proof}
  We show that $C\uplus I\uplus S$ is a valid partition of $(G, I_{0}, k)$ if and only if $C\uplus ((I\setminus I_{0})\cup I_{0}')\uplus S$ is a valid partition of $(G', I_{0}', k)$.
  Note that
  \[
    |E((I\cup S)\setminus I_{0}, I_{0})| = \sum_{x\in (I\cup S)\setminus I_{0}} |N_{G}(x)\cap I_{0}| =
    \sum_{x\in (I\cup S)\setminus I_{0}'} |N_{G'}(x)\cap I_{0}'|
    = |E((I\cup S)\setminus I_{0}', I_{0}')|.
  \]
  Since $G - I_{0} = G' - I_{0}'$, and there is no edge in $G[I_{0}]$ or $G'[I'_{0}]$, we conclude the proof.
\end{proof}

\begin{reduction}\label{rule:alternative(pseudo-split)}
  Let $(G, I_{0}, k)$ be an annotated instance where $I_{0}$ is an independent set, and let $v$ be a vertex in $C_{M}$.  
  If $v$ is not contained in any $2K_2$ or any $I_{0}$-centered $P_{3}$, and every $C_4$ and $C_5$ that contains $v$ intersects $I_{0}$, then remove $v$ from $G$.
\end{reduction}
\begin{proof}[Safeness of Rule~\ref{rule:alternative(pseudo-split)}.]
  We show that $(G, I_{0}, k)$ is a yes-instance if and only if $(G - v, I_{0}, k)$ is a yes-instance, by establishing a sequence of equivalent instances.
  For each edge $x y$ with $x\in I_{0}$ and $y\in V(G)\setminus I_{0}$, introduce a new vertex $v_{x y}$ and make it adjacent to $y$.  
  Remove all vertices in $I_{0}$, and let $I_{0}'$ denote the set of new vertices.  
  Let $(G', I_{0}', k)$ denote the resulting instance.  
  The equivalence between $(G', I_{0}', k)$ and $(G, I_{0}, k)$ follows from Lemma~\ref{lem:i-vertices(pseudo-split)}.  
  Then let $(G'', k)$ denote the graph obtained by applying Rule~\ref{rule:clean-up(pseudo-split)} to $(G', I_{0}', k)$, with $K$ being the added clique.

  We argue that $v$ is not contained in any $2K_2$, $C_4$ or $C_5$ of $G''$.
  Suppose for contradiction that there is a set $F\subseteq V(G'')$ that contains $v$ and induces a $2K_2$, $C_4$, or $C_5$ in $G''$.
  Since neither the transformation from $G$ to $G'$ nor the transformation from $G'$ to $G''$ makes any change to $V(G)\setminus I_{0}$, this set induces the same subgraph in $G$ and $G''$.  
  Thus, $F\not\subseteq V(G)\setminus I_{0}$. 
  Moreover, since every vertex in $K$ is universal in $G'' - I_{0}'$, it follows that $F\cap I_{0}'$ is not empty.  
  Note that every vertex in $I_{0}'$ has degree one in $G''$, we can conclude that $G''[F]$ must be $2 K_{2}$ and $F\cap K=\emptyset$.  
  But then $v$ is contained in either a $2 K_{2}$ or an $I_{0}$-centered $P_{3}$ in $G$, a contradiction.

  It then follows from Lemma~\ref{lem:safe-vertex-pseudo-split} that $(G'', k)$ is equivalent to $(G''- v, k)$.  
  To see the equivalence between $(G''- v, k)$ and $(G-v, I_{0}, k)$, we apply the reversed operations from $G$ to $G''$.  
  We first use Rule~\ref{rule:decided-vertices(pseudo-split)}, applied to $(G''- v, \emptyset, k)$, to mark all vertices in $I_{0}'$, then use Lemma~\ref{lem:i-vertices(pseudo-split)} to replace $I_{0}'$ by $I_{0}$, and finally remove vertices in $K$.  
  The resulting graph is precisely $G - v$.  
  We can thus conclude the proof.
\end{proof}

We call an annotated instance \emph{reduced} if none of Rules~\ref{rule:simplicial(pseudo-split)}--\ref{rule:alternative(pseudo-split)} is applicable to this instance.
\begin{lemma} \label{lemma:bound(pseudo-split)}
  If a reduced instance $(G, I_{0}, k)$ is a yes-instance, then $|C_{M}| \le 3k\sqrt{2 k}+2k+2$ and $|I_{M}|\le k+3$.
\end{lemma}
\begin{proof}
  Let $E_{-}$ be any solution to $(G, I_{0}, k)$ with at most $k$ edges and $C\uplus I\uplus S$ a pseudo-split partition pf $G-E_-$.
  Since Rule~\ref{rule:alternative(pseudo-split)} is not applicable, every vertex in $C_{M}$ is contained in some $2K_2$ or $I_{0}$-centered $P_{3}$, or some $C_4$ or $C_5$ in $G - I_{0}$.  
  If $\mathrm{opt}(G) =\mathrm{sed}(G)$, then every minimal forbidden structure contains an edge in $E_-$.
  If $\mathrm{opt}(G) < \mathrm{sed}(G)$, then all minimal forbidden structures except for one $C_5$ (it is $S$) contain an edge in $E_-$, then $|S\cap C_M|\leq 2$ (by Proposition~\ref{lem:induced-c5}). 
 Therefore, to bound $|C_{M}|$, it suffices to count how many vertices in $C_{M}$ can form a $2K_2$ or $I_{0}$-centered $P_{3}$, or a $C_4$ or $C_5$ in $G - I_{0}$ with an edge $xy\in E_{-}$.
  \begin{itemize}
  \item  If a vertex $v\in C_{M}$ is in a $2 K_{2}$ with edge $x y$, then either $v\in \{x, y\}$ or $v$ is adjacent to neither $x$ nor $y$.  In the first case, no other vertex in $C_{M}$ can occur in any $2 K_2$ with $x y$.
  Since $x y\in E(G)$, at least one of them is not in $I_{0}$ (Rule~\ref{rule:i-edges(pseudo-split)}).  
  This vertex has at most $\sqrt{2 k}+1$ non-neighbors in $C_{M}$.  
  Therefore, the total number of vertices in $C_{M}$ that can occur in any $2 K_2$ with $x y$ is at most $\sqrt{2 k}+1$.
  \item If $x y$ is an edge in any $I_{0}$-centered $P_{3}$, then precisely one of them is in $I_{0}$.  
  Assume without loss of generality $x\in I_{0}$.  
  If a vertex $v\in C_{M}$ is in an $I_{0}$-centered $P_{3}$ with the edge $xy$, then either $v = y$, or $v$ is not adjacent to $y$. 
  If $v=y$, then there is no other vertex in $C_M$ that can be in an $I_0$-centered $P_3$ with $x y$.    
  Since $y\not\in I_{0}$, it has at most $\sqrt{2 k}+1$ non-neighbors in $C_{M}$.  
  Thus, the total number of vertices in $C_{M}$ that can occur in any $I_{0}$-centered $P_{3}$ containing $x y$ is at most $\sqrt{2 k} + 1$.
  \item  If a vertex $v\in C_{M}$ is in a $C_{4}$ or $C_{5}$ that contains $x y$, then $v$ is adjacent to at most one of $x$ and $y$.  
  Since this $C_{4}$ or $C_{5}$ is in $G - I_{0}$, then each of $x$ and $y$ has at most $\sqrt{2 k}+1$ non-neighbors in $C_{M}$.  
  Thus, the total number of vertices in $C_{M}$ that can occur in such a $C_4$ or $C_5$ is at most $2 \sqrt{2 k}+2$.
  \end{itemize}
Noting that an edge cannot satisfy both the second ($|\{x, y\}\cap I_{0}| = 1$) and third ($|\{x, y\}\cap I_{0}| = 0$) categories, we can conclude $|C_{M}| \le k(\sqrt{2 k} + 2 \sqrt{2 k}+2)+2 = 3 k \sqrt{2 k}+2k+2$.

  Since Rule~\ref{rule:simplicial(pseudo-split)} is not applicable, no vertex in $I_{M}$ is simplicial.
  Let $E_-$ be a set of at most $k$ edges such that $G-E_-$ is a pseudo-split graph and $C\uplus I\uplus S$ is a pseudo-split partition of $G-E_-$.
  Since $C$ is a clique, for each vertex $v\in I_{M}\cap I$, at least one neighbor of $v$ is in $I\cup S$.  
  Therefore, each vertex $v\in I_{M}\cap I$ is incident to an edge in the solution $E_-$.  
  Noting that $I_{M}$ is an independent set, we have
  $k \ge |I_{M}\cap I| \ge |I_{M}| - 1$, where the second inequality follows from Lemma~\ref{lem:pseudo-split partition}(i). 
  Note that $|I_M\cap S|\leq 2$ by Lemma~\ref{lem:pseudo-split partition}(ii). 
  Thus, $|I_{M}| \le k + 3$, and this concludes this proof.
\end{proof}

We use the algorithm described in Figure~\ref{fig:alg-split}.
The analysis of the kernels is the same as that of Theorem~\ref{thm:split}.
\begin{theorem}
  There is an $O(k^{1.5})$-vertex kernel for the pseudo-split edge deletion problem.  
    There is a kernel of $O(k^{1.5})$ vertices and $O(k^{2.5})$ edges for the pseudo-split completion problem.  
\end{theorem}

\bibliographystyle{plainurl}

\end{document}